\documentclass[sigconf]{acmart_modified}
\usepackage{booktabs} % For formal tables

\setcopyright{none}
\newcommand{\ignore}[1]{}
\newcommand{\notinproc}[1]{#1}
\newcommand{\onlyinproc}[1]{}

\onlyinproc{
\acmDOI{10.475/123_4}

\acmISBN{123-4567-24-567/08/06}

\acmConference[KDD'17]{ACM}{August 2017}{ USA} 
\acmYear{2017}
\copyrightyear{2017}
\acmPrice{15.00}
}

\usepackage{url}
\usepackage{graphicx}
\usepackage{algorithm}
\usepackage[algo2e, ruled, vlined]{algorithm2e}

\ignore{
\newtheorem{thm}{Theorem}[section]
\newtheorem{theorem}{Theorem}[section]

\newtheorem{lemma}[thm]{Lemma}

\newtheorem{corollary}[thm]{ Corollary}

}

\DeclareMathOperator{\Laplace}{\mathcal{L}}
\DeclareMathOperator{\BinM}{\mathcal{B}}
\DeclareMathOperator{\LapM}{\mathcal{L}^{c}}
\DeclareMathOperator{\dCount}{\textsf{Distinct}}
\DeclareMathOperator{\MdCount}{\textsf{MxDistinct}}
\DeclareMathOperator{\MxDCount}{\textsf{MxDistinct}}
\DeclareMathOperator{\TdCount}{\text{{\sf TDistinct}}}
\DeclareMathOperator{\SUM}{\textsf{Sum}}
\DeclareMathOperator{\MAX}{\textsf{Max}}
\def\Cap{\textsf{cap}}
\def\softCap{\widetilde{\textsf{cap}}}

\def\loss{\textsf{loss}}
\def\relerr{\textsf{relerr}}
\def\Exp{\textsf{Exp}}
\newcommand\E{\textsf{E}}

\begin{document}
\title{HyperLogLog Hyper Extended:\\ Sketches for Concave Sublinear
  Frequency Statistics}
%\subtitle{KDD '17 submission \#452}
% \title{Estimating Frequency Statistics \\ through Distinct Count Measurements}

\author{Edith Cohen}
% \authornote{Dr.~Trovato insisted his name be first.}
% \orcid{1234-5678-9012}
\affiliation{%
  \institution{Google Research, USA}
\institution{Tel Aviv University, Israel}
%  \streetaddress{P.O. Box 1212}
%  \city{Dublin} 
%  \state{Israel} 
%  \postcode{43017-6221}
}
\email{edith@cohenwang.com}

%\author{
%Edith Cohen\\
%Google Research, USA\\
%Tel Aviv University, Israel\\
%\texttt{edith@cohenwang.com}
%}

%\date{}

\begin{abstract}
 
  One of the most common statistics computed over data elements is the
  number of distinct keys. 
A thread of research pioneered by
  Flajolet and Martin three decades ago culminated in the design of 
optimal approximate counting sketches, which have size that is
  double logarithmic in the number of distinct keys and provide estimates
  with a small relative error.  Moreover, the sketches are
  composable, and thus suitable for streamed,
  parallel, or distributed computation.

 We consider here all statistics of the frequency distribution of keys, where a contribution
  of a key to the aggregate is concave and grows (sub)linearly with
  its frequency.  These fundamental aggregations are very
  common in text, graphs, and logs analysis and include logarithms,
  low frequency moments, and capping statistics.

We design composable sketches of double-logarithmic size for
all concave sublinear statistics.  
Our design combines theoretical
optimality and practical simplicity. 
 In a nutshell, 
we specify tailored
  mapping functions of data elements to output elements so that
our target statistics on the data elements is approximated by the
  (max-) distinct statistics of the output elements, which can be
  approximated using off-the-shelf sketches.
Our key insight is relating these target statistics  to the
{\em complement Laplace} transform of the input frequencies.

\end{abstract}

\maketitle
\section{Introduction}

\SetKwData{Okeys}{OutKeys}
\SetKwData{Oelems}{OutElements}

  We consider data presented as elements $e = (\text{\em{e.key}},\text{\em{e.value}})$ where
  each element has  a {\em key}  and
 a positive numeric {\em value} $ > 0$.  
 This data model is very common  in streaming or distributed 
  aggregation problems. 
A   well-studied special case is where $\text{\em{e.value}} \equiv 1$ for all
  elements.

  One of the most fundamental statistics over such data is the number
  of distinct keys: $\dCount(E) = \{\text{\em{e.key}}\mid e\in E\}$.  Exact computation of the statistics requires
  maintaining a structure of size that is linear in the number of distinct keys.
A pioneering design of Flajolet and Martin \cite{FlajoletMartin85}
showed that an approximate count can be obtained in a streaming model
using structures (``sketches'') of logarithmic size.
Since then, a rich research strand
proposed and analysed a diverse set of approximate counting
sketches and deployed them
for a wide range of applications \cite{hyperloglogpractice:EDBT2013}.

Distinct counting sketches can be mostly classified as based on
sampling (MinHash sketches) or on random
  projections (linear sketches).
Both types of structures are mergeable/composable:
This means that when the elements are partitioned, we can compute a
sketch for each part separately and then obtain a corresponding sketch
for the union from the sketches of each part.
This property is critical for making the sketches
suitable for parallel or distributed aggregation. 

  The original design of \cite{FlajoletMartin85} and the leading
ones used in practice use sample-based sketches.  In particular, 
the popular Hyperloglog sketch \cite{hyperloglog:2007} 
has double logarithmic size $O(\epsilon^{-2}+\log\log n)$, where $n$ is the
number of distinct keys and $\epsilon$ is the target
normalized root mean squared error (NRMSE).
Since this size is necessary to represent the approximate
  count, Hyperloglog is asymptotically optimal.
We note that  the Hyperloglog sketch contains
$\epsilon^{-2}$ registers which store exponents of the estimated
count.  Thus, explicit representation of the sketch has 
size $O(\epsilon^{-2}\log\log  n)$, but one can theoretically use instead a single exponent
and $\epsilon^{-2}$ constant-size offsets 
  (e.g. \cite{KNW:PODS2010,ECohenADS:TKDE2015}) to  bring the sketch
  size down to $O(\epsilon^{-2}+\log\log n)$, albeit by
  somewhat increasing updates complexity.
% that are constant size in expectation with concentration.
Another point is that Hyperloglog uses random hash
functions which have logarithmic-size representations.  Some
theoretical lower bounds consider the hash representation 
to be part of the sketch \cite{ams99,KNW:PODS2010}), implying a
logarithmic lower bound on size.
 We follow \cite{FlajoletMartin85,hyperloglog:2007}  and consider the hash 
  representation to be provided by the platform, which is consistent with practice where hash functions are
reused and shared by multiple sketches.

  We now consider other common statistics over elements.  In particular,
  statistics expressed over a set of (key, weight) pairs, where the weight
  $w_x$ of a  key $x$,  is defined to be the sum of the values of data
  elements with key $x$: $$w_x = \sum_{e \mid \text{\em{e.key}}=x} \text{\em{e.value}}\
  .$$   Keys that are not active (no elements with this key) are 
  defined to have $w_x = 0$.   Note that if all elements have value
  equal to $1$, then
  $w_x$ is the number of occurrences of key $x$.
  For a nonnegative function $f(w) \geq 0$ such that $f(0)\equiv 0$,  we define the
  {\em $f$-statistics} of the data as $\sum_x f(w_x)$.  We will find
  it convenient to work with the notation $W(w)$ for the number of
  keys with $w_x = w$.   Equivalently, we can treat
  $W$ as a distribution over weights $w_x$  that is scaled by the number of
  distinct keys.  We can then express the $f$-statistics (with a slight notation abuse) as 
\begin{equation}f(W)   = \int_0^\infty  W(w)f(w) dw\ .\end{equation}

 The design og sketches that approximate $f$-statistics over streams of elements was
  formalized and popularized in a seminal paper by Alon, Matias, and Szegedy
  \cite{ams99}.   The aim was to understand the necessary sketch size
  for approximating different statistics and in particular understand
  which statistics can be sketched in size that is polylogarithmic in the number   of keys.

 We focus here on functions $f$ that are
concave with (sub)linear nonnegative
growth-rate.  Equivalently, these are all
nonnegative linear combinations (the nonnegative
{\em span}) of {\em capping} functions
$$\Cap_T(x)   = \min\{T,x\}\    \text{parameterized by}\ T>0\ ,$$
and we therefore use the notation $\overline{\Cap}$.
Members of $\overline{\Cap}$ that parametrize a spectrum
between distinct count ($f(x)=1$) and sum ($f(x)=x$) include
  {\em frequency moments} $f(x) =
  x^p$ in the range $p=[0,1]$ (sum is $p=1$ and distinct count is $p=0$),  capping functions (sum is realized by $\Cap_\infty$ and distinct count by $\Cap_1$ when element
  values are integral and by $f(x)=\Cap_T/T$ as $T\rightarrow 0$
  generally), and   the {\em soft capping}
functions
\begin{equation} \label{softcapdef:eq} \softCap_T = T(1-\exp(-x/T))\ .\end{equation}
Soft 
capping is a smooth approximation of ``hard'' capping functions: For 
$x\ll T$ we have $\softCap_T(x) \approx x = \Cap_T(x)$,  for 
$x \gg T$ we have $\softCap_T(x) \approx T = \Cap_T(x)$, and it always
holds that
\begin{equation}\label{softcap:eq}
\forall x,\, (1-1/e) \Cap_T(x) \leq  \softCap_T(x) 
\leq \Cap_T(x)\ . 
\end{equation}
Other important  $\overline{\Cap}$ members are $\log(1+x)$ and capped moments.

Statistics in 
$\overline{\Cap}$ are used in applications to decrease the impact of very 
frequent keys and increase the impact of rare keys.  
It is a common practice to weigh frequencies, say degree of
nodes in a graph \cite{NandanwarM:kdd2016} or frequency of a term in a document~\cite{SaltonBuckley1988}, by a sublinear
function such as $w^p$ for $p\in (0,1)$ or $\log(1+w)$.  
In many applications, the ability to approximate the statistics over the raw data, without the cost 
of aggregation, can be very useful.
One such application is to online advertising 
\cite{GoogleFreqCap,facebookFreqCap}, data elements are 
interpreted as opportunities to show ads to users (keys) that are interacting with various apps on different 
platforms.  
An advertisement campaign specifies a maximum number of 
times $T$ an ad can be displayed  to the same user, so the number of 
qualifying impressions corresponds to $\Cap_T$ statistics of the data. 
Statistics are computed over past data in order to estimate the number 
of qualifying impressions when designing a campaign. 
Another application is the computation of
word embeddings.   The objective  is to have the focus and
context embeddings of words captures the respective co-occurrence
frequencies.  Glove~\cite{PenningtonSM14:EMNLP2014} achieved
significant improvements by using 
$f(w) = \min\{1,w/T\}^\alpha$ for $\alpha<1$ (instead of $f(w)=w$).  Typically, 
the text corpus is presented as complete text documents, and elements in arbitrary order are extracted in a pass over
this data.

  There is a very large body of work on the topic of approximating
  statistics over streamed or distributed data and it is not
  possible to mention it all here.  Most of the prior work uses linear sketches (random linear projections).  A sketch for the second moment, inspired by the JL transform \cite{JLlemma:1984}, was presented by \cite{ams99}.
  Indyk \cite{indyk:stable} followed with
a beautiful construction based on stable distributions
of sketches of size $O(\epsilon^{-2}\log^2 n)$ for
moments in $p\in [0,2]$. 
Braverman and Ostrovsky \cite{BravermanOstro:STOC2010} presented a
characterization and an
umbrella construction of sketch structures, based on $L_2$ heavy
hitter sketches, for all monotone $f$-statistics.  The structure size
is polylogarithmic  but is
practically too large (degree of the
polylog and constant factors).

Sample-based sketches for $\overline{\Cap}$ functions were presented
by the author~\cite{freqCap:KDD2015}.  
The sketch is a weighted sample of keys that supports
approximate $\Cap$-statistics on domain queries (subsets
of the keys).   The framework generalizes both distinct reservoir
sampling \cite{Knuth2f,Vit85} and the sample and hold stream sampling 
 \cite{GM:sigmod98,EV:ATAP02,flowsketch:JCSS2014}. The size and
quality tradeoffs of the sample 
are very close (within a small constant) to those of an optimal sample that can be efficiently
computed over aggregated data (set of key and weight pairs).  Roughly,
a sample of $O(\epsilon^{-2})$ keys suffices to approximate
$\Cap_T(W)$ unbiasedly with coefficient of variation (CV) $\epsilon$.  
Moreover, a {\em multi-objective} (universal)  sample (see \cite{multiw:VLDB2009,multiobjective:2015}) of 
$O(\epsilon^{-2} \log n)$ keys can approximate with 
CV $\epsilon$ {\em any}  $f$-statistics for $f\in \overline{\Cap}$.
When this method is applied to sketching statistics of the full data, we can hash
key identifiers to size $O(\log n)$ (to obtain uniqueness with very
high probability) and obtain sketches of size $O(\epsilon^{-2} \log
n)$ and a multi-objective sketches of size $O(\epsilon^{-2} \log^2
n)$.  One weakness of the design is that these sketches
are not fully composable: They apply on streamed elements (single
pass) or take two passes over distributed data elements. 

A remaining fundamental challenge is to
design composable sketches  of size $O(\epsilon^{-2}
\log n)$ for each $\overline{\Cap}$ statistics and a
composable multi-objective sketch of size $O(\epsilon^{-2}
\log^2 n)$.  Given the practical significance of the problem, we seek
simple and highly efficient designs. A further theoretical challenge
is to design sketches that meet or approach the double-logarithmic representation-size 
lower bound of $O(\epsilon^{-2}+ \log \log n)$.

\ignore{
\begin{quote}
{\bf Q1: } 
Can we approximate all statistics in $\overline{\Cap}$ using 
structures of size $O(\epsilon^{-2} + \log \log n)$ ?
\end{quote}

\begin{quote}
{\bf Q2: } 
Can we design composable structures for estimating $\overline{\Cap}$
statistics of size that is logarithmic  in $n$ ?
\end{quote}
Linear sketches are fully composable but current structures have
larger sizes.
The sampling-based framework of \cite{freqCap:KDD2015} included a streaming
scheme and a two-pass composable scheme (the first pass
determines the keys in the sample and another computes their exact
weights). But we seek a one-pass composable scheme.
}

\subsection*{Contributions overview and organization}
We address these challenges and make the following contributions.
We show that
{\em any} statistics in the {\em soft capping span}
    $\overline{\widetilde{\Cap}}$ can be approximated with the essential
  effectiveness and estimation quality of Hyperloglog.
That is, we present composable sketches of size
$O(\epsilon^{-2}+\log\log n)$ that estimate the statistics with
RNMSE $\epsilon$ with good concentration.
The soft capping span $\overline{\widetilde{\Cap}} \subset \overline{\Cap}$
is the set of functions that can be expressed as
\begin{equation} \label{softspanf:eq}
f(w) = \int_0^\infty a(t)
(1-e^{-wt}) dt \ , \text{ where } a(t)\geq 0\ .
\end{equation}
In particular, the span includes 
all soft capping functions, low
frequency moments ($f(w)=w^p$ with $p\in (0,1)$), and $\log(1+w)$.
We also present a composable {\em  multi-objective} sketch for
$\overline{\widetilde{\Cap}}$.  This is a single structure that is
larger by a logarithmic factor than 
a single distinct counter and supports the approximations of
{\em all} $\overline{\widetilde{\Cap}}$ statistics.
Finally, we consider statistics in
$\overline{{\Cap}}$ that are not in $\overline{\softCap}$ and show how
to approximate them within small relative errors (12\%) using
differences of approximate
$\overline{\softCap}$ statistics.  

 Our main technical tool is a novel
 framework, illustrated in Figure~\ref{framework:fig},
 that reduces the approximation of the more general
 statistics to distinct-count statistics.  At the core we
 specify randomized functions $M$ that 
map data elements of the form $e=(\text{\em{e.key}},\text{\em{e.value}})$ to
sets of {\em output elements}.  Each output element $e'\in M(e)$
contains an
{\em output key} ({\em outkey}) $e'.key$ (which generally is from a
different domain than the input keys)  and an optional value $e'.value
\geq 0$.  For a set of data elements $W$, we obtain a
corresponding multiset of output elements $$E=M(W)= \bigcup_{e\in W} M(e)\ .$$
The mapping functions are crafted so that
the approximate statistics of
the set of data elements $W$ can be obtained from
approximations of other statistics of the output
elements $E$.  In particular, if we have a composable sketch for
the output statistics,  we obtain a composable sketch for the target
statistics.
Note that our mapping functions are randomized, therefore the set $E$
is a random variable and so is any (exact) statistics on $E$.  We will
use approximate statistics on $E$.

  The output statistics we use are the distinct count
  $\dCount(E)$, which allows us to leverage 
approximate distinct counters as black boxes, and the
more general 
{\em max-distinct} statistics $\MdCount(E)$, defined as
the sum over distinct keys of the maximum value of an element with the
key: 
\begin{eqnarray} \label{mxdistinct:eq}
\lefteqn{\MdCount(E) = \sum_x m_x\ ,}  \label{mxdistinct:eq}\\
&& \text{ where } m_x\equiv\max_{e\in E \mid \text{\em{e.key}}=x}
\text{\em{e.value}}\ ,\nonumber
\end{eqnarray}
 which also can be sketched in double logarithmic size.
Note that when all elements have value $1$, $\MdCount(E) = \dCount(E)$.

For multi-objective approximations we use {\em all-threshold}
sketches that allows us to recover, for any threshold $t>0$,  an
approximation of  
\begin{equation} \label{tdistinct:def}
\TdCount_t (E)=\dCount \{e\in E \mid \text{\em{e.value}} \leq t \} ,
\end{equation}
 which is the number of distinct keys that appear in at least one
element $e\in E$ with value $\text{\em{e.value}} \leq t$.
The size of the all-threshold sketch is
larger by only a logarithmic factor than the basic distinct count
sketch.
We refer to the value of the output statistics on the output elements
as a {\em measurement} of $W$.  When a sketch is applied to approximate the value of output statistics, we refer to the result as
an {\em approximate measurement} of $W$.

% Terminology wise, we refer to $\dCount(M(W))$ (or $\MdCount(M(W))$) as a measurement of $W$ by $M$.
% When an approximate counter is used we refer to the result as an {\em approximate measurement}. 

\begin{figure}
\center
\includegraphics[width=0.45\textwidth]{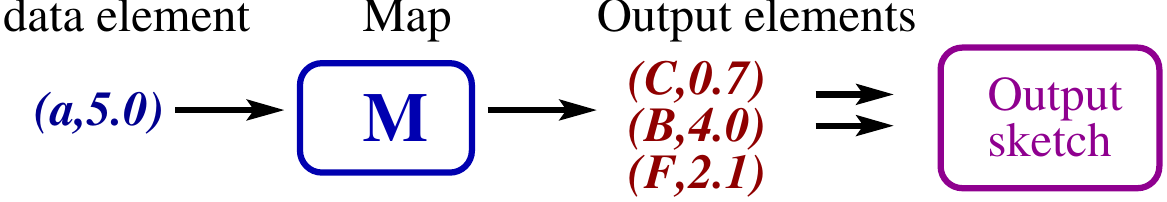}
\caption{{\small Element processing framework}}
\label{framework:fig}
\end{figure}

\ignore{
\begin{algorithm}[h]\caption{Base framework\label{framework:alg}}
 {\small
\KwIn{Data elements $W$, function $f \in \overline{\Cap} $}
\KwOut{Estimate of $f(W)$}
Specify (tailored to $f$):
\begin{itemize}
  \item
    Mapping functions $\{M_i\}$ of data elements to sets of output elements
  \item  
    An estimator $\hat{f}:R^{|\{M_i\}|}$
\end{itemize}    
\ForEach{$i$}{$\hat{C}_i \gets$ approximate (max-)distinct count of $M_i(W)$}\;
\Return $\hat{f}(\{\hat{C}_i\})$
}
\end{algorithm}
}

The paper is organized as follows.
In Section~\ref{cLaplace:sec} we define the {\em complement Laplace transform}
$\LapM[W](t)$  of the frequency
distribution $W$, which is its distinct count minus its Laplace
transform at $t$.  We have the relation
$$T \LapM[W](1/T) = \softCap_T(W)\ ,$$ that
is, the transform at $1/T$ multiplied by $T$ is the $\softCap_T$ statistics of
the data.
In Section~\ref{approxcL:sec} we define a mapping function for any
$t>0$,  so that $\LapM[W](t)$, and hence  $\softCap_{1/t}$-statistics, is
approximated by the respective $\dCount$ measurement.
We refer to this as a
measurement of $\LapM[W]$ at point $t$.

In Section~\ref{softspan:sec} we 
consider the span $\overline{\softCap}$ of soft capping statistics,
that is, all $f$ of the form \eqref{softspanf:eq}.
Equivalently,  $a(t)$ is the inverse $\LapM$ transform of $f$.  
We derive the explicit form of the inverse transform of
all frequency moments with $p\in(0,1)$ and logarithms.
 The statistics $f(W)$ for  $f\in \overline{\softCap}$ can thus be expressed as 
$$f(W) = \LapM[W][a] \equiv \int_0^\infty a(t) \LapM[W](t) dt\ .$$ 
This suggests that we can approximate $f(W)$ using multiple approximate point ($\dCount$)
measurements. In section~\ref{weightedDC:sec} we show
that a single
$\MdCount$ measurement suffices:
We present
 element mapping functions (tailored to $f$) such that the
 $\MdCount$ statistics on output elements approximates $f(W)= \LapM[W][a]$.
 We refer to this statistics as a {\em combination} measurement
 of $\LapM[W]$ using $a$. A $\MdCount$ sketch of the output element gives us
  an approximation of  combination measurement which approximates $f(W)$.
Finally, we will review the design of HyperLoglog-like $\MdCount$ sketches.

In Section~\ref{allt:sec} we consider the multi-objective setting, that is, a single sketch from which we can approximate all $\overline{\softCap}$ statistics.
We define a mapping function such that for all $t>0$, $\TdCount_t (E)$ is equivalent to a point measurement of $\LapM[W]$ at $t$.
\ignore{
{\em full-range $\LapM[W]$
measurement,} which encodes point measurements of 
$\LapM[W](t)$ at all points $t$.  This measurement provides us with
simultaneous approximations for all functions in $\overline{\softCap}$.
We show how to compute and represent an approximate full-range measurement
within a logarithmic overhead over that of a single approximate point
measurement using an all-threshold sketch.

To do so, we define
the {\em all-threshold} count  of a set of data elements $E$ that
are key and value pairs.  The count $\TdCount_t (E)$ 
for $t > 0$ is the number of distinct keys that appear in at least one
element $e$ with value $\text{\em{e.value}} \leq t$.
The full-range measurement is defined using our framework where
output elements have values and are processed by an all-threshold
distinct counter.
}
The output elements are processed by an
{\em all-threshold distinct count} sketches, which
can be interpreted as all-distance sketches
\cite{ECohen6f,ECohenADS:TKDE2015} and inherit their properties -- In
particular, the total structure size has logarithmic overhead over a
single distinct counter.  The all-threshold sketch
allows us to obtain approximate point measurements for any $t$ and combination measurement for any $a$.

In Section~\ref{hardcappingtransform:sec} we consider statistics in
$\overline{\Cap}$ that may not be  in $\overline{\softCap}$.
We characterize $\overline{\Cap}$ as the set of all concave sublinear
functions and derive expressions for the {\em capping 
  transform} which transforms $f\in \overline{\Cap}$ to the 
coefficients of the corresponding nonnegative linear combination of 
capping functions.  
We then consider sketching these statistics $f(W)$ using
approximate {\em signed} inverse $\LapM$
transform of the function $f$.  We use
separate combinations
measurements of the positive and negative components for the
approximation.    We show that $\Cap_1$ is the ``hardest''
function in that class in the sense that 
any approximate inverse transform for 
the function $\Cap_1(x) = \min\{1,x\}$ can be extended (while
retaining sketchability and approximation quality) to
any statistics $f\in \overline{\Cap}$,  using
the capping transform of $f$.  We then derive some approximate
transforms for $\Cap_1(x)$, and hence for any $\overline{\Cap}$
statistics  that achieve maximum relative error of $12\%$.

\notinproc{
Section~\ref{experiments:sec} shows some experiments that 
mainly aimed at demonstrating the ease and effectiveness of our sketches.

Appendix sections \ref{bintrans:sec}--\ref{weighted:sec} briefly present some
 extensions.
In Section~\ref{bintrans:sec} we define a related 
discrete transform which we call the {\em Bin transform}.  A variant of the 
Bin transform was proposed in \cite{freqCap:KDD2015}.  We show how it 
can be approximated using our framework. 
Sections~\ref{decay:sec} and~\ref{weighted:sec}  discuss extensions to
time-decayed statistics and weighted keys.
We conclude in Section~\ref{conclu:sec} with future directions and
open problems.
}
\onlyinproc{
We conclude in Section \ref{conclu:sec}.  Due to page limitations,
many proofs and details and also some experimental results are omitted from this submission.  A full
version can be found in \url{https://arxiv.org/abs/1607.06517}.
}
%  We will  obtain a better
% approximation using $\LapM[W](t)$ measurements at three points.

\section{The Laplace$^C$  transform} \label{cLaplace:sec}
The {\em complement Laplace (Laplace$^c$) transform} $\LapM[W](t)$ is parametrized by $t>0$ and
defined as
\begin{eqnarray}
\lefteqn{\LapM[W](t) \equiv  \int_0^\infty W(w) (1-\exp(-wt))dw } \nonumber\\
&=&  \int_0^\infty W(w)dw-\Laplace[W(w)](t) \label{termsL:eq}\ .
\end{eqnarray}
See Figure~\ref{example:fig} for an illustration of $LapM[W](t)$  for a toy 
distribution $W$. 
The transform has the following relation to 
soft capping statistics:
% \begin{lemma} \label{transform:lemma}
\begin{eqnarray}\label{laplacecap:eq}
\softCap_T(W) &=& T \LapM[W](1/T) 
% T \LapM[W](1/T) & \leq&    \Cap_T(W)  \leq \frac{e}{e-1}  T 
%                                       \LapM[W](1/T) \ . 
\end{eqnarray}
% \end{lemma}
% \begin{proof}
% The first equality is immediate from the definitions.  The second 
% follows from \eqref{softcap:eq}: $$\forall w\geq 0, 
% (1-1/e)\min\{1,wt\} \leq (1-\exp(-wt) \leq \min\{1,wt\}\ .$$
% \end{proof}

The first term in \eqref{termsL:eq}, $\int_0^\infty W(w)dw \equiv \dCount(W)$, is the ``distinct count''  and the
second term $\Laplace[W(w)]$ is the Laplace transform of our (scaled) frequency 
distribution $W$. Hence the name {\em complement Laplace} transform.
Note that $\LapM[W](t)$ is non-decreasing with 
  $t$. At the limit when $t$ increases, the second term
  vanishes and 
\begin{equation} \label{hight:eq}
\lim_{t \rightarrow \infty} \LapM[W](+\infty) =
  \int_0^\infty W(w)dw = \dCount(W)
\end{equation}
 is the number of 
  distinct keys in $W$.   At the limit as $t$ decreases
\begin{equation} \label{lowt:eq}
\lim_{ t \rightarrow 0^+} \frac{1}{t}\LapM[W](t)
  = \int_0^\infty W(w) w dt = t \SUM(W)\ ,
\end{equation}
where $\SUM(W) = \sum_{e\in W} \text{\em{e.value}} = \sum_x w_x$ is the sum of the
weights of keys.  More precisely:
\begin{lemma} \label{relevantrange:lemma}
For $t\leq \frac{\sqrt{\epsilon}}{\max_x w_x}$ and for
$t\geq \frac{-\ln \epsilon}{\min_x w_x}$, the transform is approximated
within a relative error of at most $\epsilon$ by the respective limits
\eqref{lowt:eq} and~\eqref{hight:eq}.
\end{lemma}
\begin{proof}
For the first claim, note that $wt\leq \sqrt{\epsilon}$.
Hence, using the Maclaurin expansion, $|1-e^{-wt}-wt| \approx
  (wt)^2/2 \leq \epsilon$.
For the second claim, the relative error is $\exp(-wt) \leq \epsilon$.
\end{proof}
The Lemma implies that the fine structure of $W$ is captured by a
restricted  ``relevant'' range of $t$ values and is well approximated 
outside this range by the basic (and composably sketchable) $\dCount$ and $\SUM$ statistics.
The statistics $\dCount(W)$ is approximated by an
off-the-shelf approximate distinct counter applied to data elements.
The exact $\SUM(W)$ is straightforward to compute 
composably with a single counter of size $O(\log \SUM(W))$ (assuming integral values). 
A classic algorithm 
by Morris~\cite{Morris77}  (see~\cite{ECohenADS:TKDE2015} for a 
composable version that can handle varying weights) uses 
sketches of size $O(\epsilon^{-2}
+ \log\log (\SUM(W))$.

\begin{figure}
\center 
\includegraphics[width=0.42\textwidth]{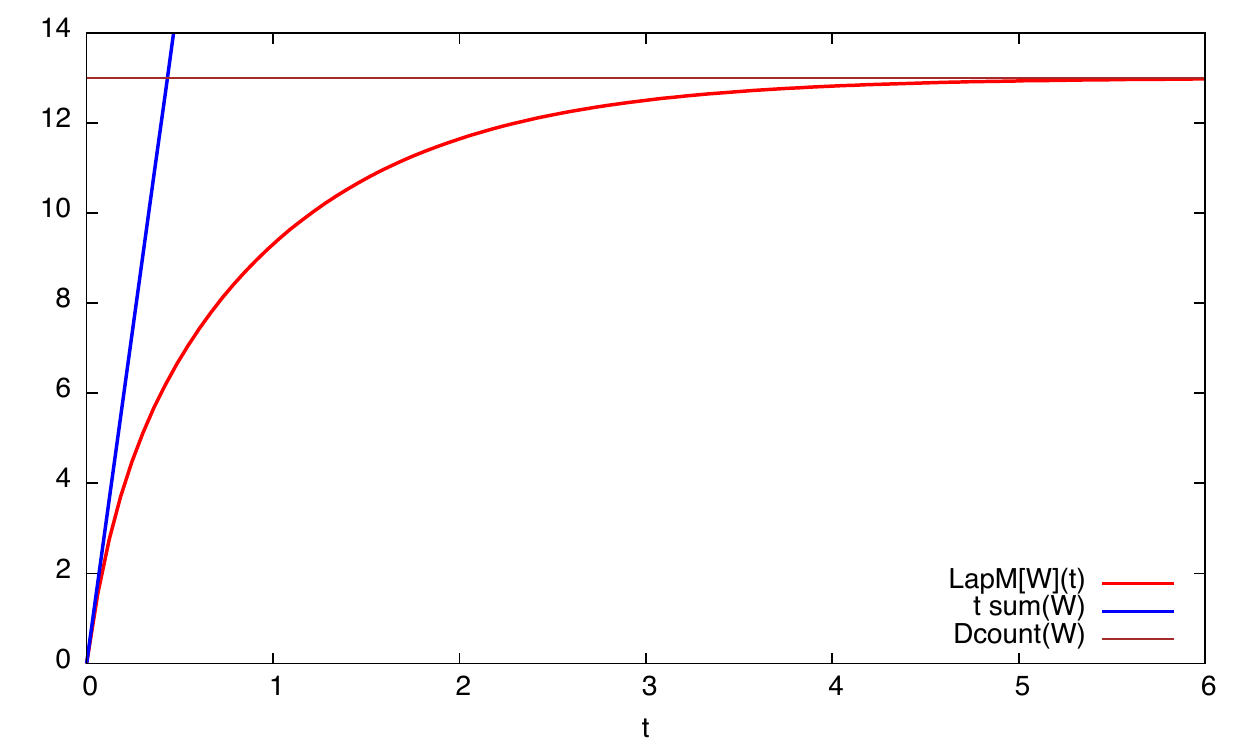}
\caption{Data $W$ with 10 keys with $w_x=1$, 2 keys with $w_x=5$, and one
  key with $w_x=10$.   The distinct count (number of keys) is
  $\dCount(W) =13$ and
  the sum is $\SUM(W) = 30$.
We have $W(w) = 10\delta(w-1)+2\delta(w-5)+
  \delta(w-10)$, where $\delta$ is Dirac Delta function.   The transform (plotted) is
$\LapM[W](t)=10(1-e^{-t})+2(1-e^{-5t})+ (1-e^{-10t}) = 13-10e^{-t}-2e^{-5t}-e^{-10t}$. The plot
shows the asymptotes $t \SUM(W)$ for small $t$ and $\dCount(W)$ for large $t$.
\label{example:fig}}
\end{figure}

\section{Laplace$^c$ point measurements} \label{approxcL:sec}
We define a mapping function of elements such that
the expectation of the (scaled) distinct count of output elements is equal to
the Laplace$^c$ transform $\LapM[W](t)$ of $W$
  at $t$. We also establish concentration around the expectation.

The basic element mapping is provided
as Algorithm~\ref{elemmap1:alg}.  A more efficient variant that
performs computation proportional to the number of generated output
elements is provided in
\notinproc{Appendix~\ref{efficientpoint:sec}.}\onlyinproc{the full version.}

The mapping is
parametrized by $t$ and by an integer $r\geq 1$ and uses a set of
functions $H_i$ for $i\in [r]$.  
All we need to assume here is that
for all $i$ and keys $x$, $H_i(x)$ are
(nearly) unique.  This can be achieved by concatenating $x$ to a
string representation of $i$: $H_i(x)
\equiv x  \cdot \text{str}(i))$.  
To obtain output key  representation that is logarithmic in $r
\dCount(W)$, we can apply a random hash function to the concatenated string.
  An element $e$ is processed by drawing a set of $r$ independent
  exponential random variables $y_i \sim \Exp[\text{\em{e.value}}]$  with parameter $\text{\em{e.value}}$.  For each $i$
  such that $y_i<t$, the output key $H_i(\text{\em{e.key}})$ is created.  Note that the
  number of output keys returned is between 0 and $r$.
  \begin{algorithm2e}[h]\caption{$\protect\Okeys_{t,H}(e)$:
      Map input
      element to  outkeys \label{elemmap1:alg}}
\DontPrintSemicolon
    {\small
\KwIn{Element $e=(\text{\em{e.key}},\text{\em{e.value}})$, $t>0 $, integer $r \geq 1$, hash
  functions $H_i$ $i\in[r]$}
\KwOut{A  set \Okeys of at most $r$ outkeys}
$\Okeys \gets []$ \tcp*[f]{initialize}\;
\ForEach{$i\in [r]$}{$y_i \sim
\Exp[\text{\em{e.value}}]$ \tcp*[f]{independent exponentially distributed with parameter $\text{\em{e.value}}$}\;
\If{$y_i \leq t$}{$\Okeys.append(H_i(\text{\em{e.key}}))$ \tcp*[f]{Append $H_i(\text{\em{e.key}})$
    to list of output keys}\;
}
}
\Return $\Okeys$
}
\end{algorithm2e}
  
% If the value is $y \leq t$, we send $\text{\em{e.key}}$ to the distinct counter.
Our point measurement is
\begin{equation}
\widehat{\LapM}[W](t) = \frac{1}{r} \dCount\left( \bigcup_e
  \text{\sc{OutKeys}}_{t,H}(e) \right)\ ,
\end{equation} which is number of distinct output keys generated for all
data elements,
divided by $r$.
We now show that for any choice of $r\geq 1$,  $t$, and input data
$W$, the expectation of the measurement $\widehat{\LapM}[W](t)$
is equal to the value of the Laplace$^c$ transform of $W$ at $t$.
\begin{lemma}  \label{unbiasedpoint:lemma}
 $$\E[\widehat{\LapM}[W](t)] = \LapM[W](t)$$
\end{lemma}
\begin{proof}
The distinct count of outkeys $r\widehat{\LapM}[W](t)$
can be expressed as the sum   of $r n$ Poisson trials.  Each Poisson trial corresponds
to ``appearance at least once''
of the potential outkey $H_i(x)$ over the $n=\dCount(W)$ active input keys and $i\in [r]$.
For each $i\in [r]$ and key $a$, the  outkey $H_i(a)$ appears
if the minimum $\Exp[\text{\em{e.value}}]$ draw over 
   elements $e$ with key $a$ is at most $t$.  The minimum of these 
   exponential random variables is 
   exponentially distributed with parameter equal to their sum $w_a = \sum_{e \mid \text{\em{e.key}}=a}
   \text{\em{e.value}}$.
 This  distribution has density 
   function $w \exp(-w x)$.  Therefore, the probability of the event is
\begin{equation}
p(w_a,t) = \int_0^t w_a \exp(-w_a y) dy = 1-\exp(-w_a t)\ . 
\end{equation}
 It follows that the expected contribution of a key $a$ with weight $w_a$ to the sum
   measurement is $rp(w_a,t)$.  Therefore
the expected value of the measurement is
$$\E[\widehat{\LapM}[w](t)]  = \frac{1}{r} \int_0^\infty W(w)r p(w,t) dw \equiv \LapM[W](t) \ .$$
\end{proof}

Note that for $t = +\infty$, which corresponds to distinct counting, $p(w,t) = 1$, and the measurement is always equal to
its expectation $n$.
More generally, $p(w,t)<1$ and we can bound the relative error
by a straightforward
   application of Chernoff bound:
\begin{lemma} \label{chernoff:lemma}
For $\delta<1$, 
$$\Pr[\frac{|\widehat{\LapM}[W](t) - \LapM[W](t)|}{\LapM[W](t)} \geq
\delta] \leq 2 \exp(-r \delta^2 \LapM[W](t)/3)\ .$$
\end{lemma}
\notinproc{
\begin{proof}
We apply Chernoff bounds to bound the deviation of the sum $r
\widehat{\LapM}[W](t)$ of our $rn$ independent Poisson
random variables from its expectation $r \LapM[W](t)$.
%   Also recall that
% $\LapM[W](t) = \sum_a p(w_a,t)$.  Thus, the expectation of the sum of all our Poisson trials
% is  $r \sum_a p(w_a,t) = r\LapM[W](t)$.  Noting that
\begin{eqnarray*}
\lefteqn{\Pr[\frac{|\widehat{\LapM}[W](t) - \LapM[w](t)|}{\LapM[W](t)} \geq
\delta]  }\\
&=&\Pr[|r\widehat{\LapM}[W](t) - r\LapM[W](t)| \geq \delta r\LapM[W](t)]\ .\nonumber
\end{eqnarray*}
\end{proof}
}

% The error can be decreased  by  replicating part of the
%  process.  We can use $r$ independent
%  hash functions $H_i$, compute independent $y_i \sim
%  \Exp[e.weight]$ $i\in[r]$ for each element, and feed $H_i(\text{\em{e.key}})$
%  when $y_i \leq t$ to the distinct counter.  (Note that we can still
%  use one distinct counter).  The error bound then decreases to
% $2 \exp(-r \delta^2 \LapM[w](t)/3)$.

The outkeys $E$ are processed by an {\em
  approximate} distinct counter and our final approximate
measurement
\begin{equation} \label{approxpointm:eq}
  \widehat{\widehat{\LapM}}[W](t) = \frac{1}{r}\widehat{\dCount}(E)
\end{equation}
is equal to the
approximate count of distinct output keys divided by $r$.
Since there are at most $r\dCount(W)$ distinct output keys, the
sketch size needed for CV $\epsilon$ is 
$O(\epsilon^{-2}+\log\log(r \dCount(W))) = O(\epsilon^{-2}+ \log\log
\dCount(W))$.
Note that even a very large choice of $r$ will not significantly
increase the sketch size.  A large $r$, however, can affect the
element mapping computation when there are many generated output
elements.

We now consider the choice of $r$ that suffices for our
quality guarantees.
From Lemma~\ref{chernoff:lemma},
when 
\begin{equation} \label{rcond:eq}
  r \LapM[W](t) \geq 3\epsilon^{-2}
  \end{equation}
we have
CV $\epsilon$ with tight concentration for
$\widehat{\LapM}[W](t)$ as an approximation of $\LapM[W](t)$.
Note that $r \LapM[W](t)$ is the expected number of distinct output
elements.
We can combine the contributions to the error due to the measurement
and its approximation, noting that the two are independent, 
 and obtain that when \eqref{rcond:eq} holds, 
the error of 
$\widehat{\widehat{\LapM}}[W](t)$ as an approximation of $\LapM[W](t)$ has
CV $\sqrt{2}\epsilon$ and tight concentration.

 We now provide a lower bound on $\LapM[W](t)$ as a function of $t$.
 \begin{lemma} \label{oversqrt:lemma}
$$\LapM[W](t) \geq \frac{e-1}{e}\SUM(W) \min\{\frac{1}{\MAX(W)},t\}\ ,$$
where $\MAX(W) \equiv \max_x w_x$.
\end{lemma}
\begin{proof}
  From definition, 
  {\small 
\begin{equation} \label{trivbound:eq}
\LapM[W](t) \geq (1-\frac{1}{e}) \sum_x \min\{1,t w_x\}
\end{equation}
  }
  When $t\leq 1/\MAX(W)$ we have  $\sum_x \min\{1,t w_x\}= t \SUM(W)$.
When $t\geq 1/\MAX(W)$ we have $\sum_x \min\{1,t w_x\}\geq \SUM(W)$.
\end{proof}
The lemma 
implies that when $t\geq \sqrt{\epsilon}/\MAX(W)$,
\begin{equation} \label{rseamless:eq}
  r=\frac{e}{e-1}\epsilon^{-2.5}\frac{\MAX(W)}{\SUM(W)}\leq \frac{e}{e-1}\epsilon^{-2.5}
\end{equation}
always suffices to ensure that 
\eqref{rcond:eq} holds.  

Since we do not know $\MAX(W)$ or $\SUM(W)$ in advance, we 
use the following strategy.
Our approximate point measurement algorithm
computes both an approximate sum $\widehat{\SUM}(W)$ and approximate
count of output elements $\widehat{\dCount(E)}$ generated by
Algorithm~\ref{elemmap1:alg} with $r$ as in \eqref{rseamless:eq}.
If $\widehat{\dCount(E)} < 3\epsilon^{-2}$, we return 
$t \widehat{\SUM}(W)$ and otherwise return \eqref{approxpointm:eq}.

 We comment here that the choice of $r$ in \eqref{rseamless:eq}
 provides seamless worst-case quality guarantees for any
 $t$ and distribution $W$.  In particular, having $t$ as small as
 $\sqrt{\epsilon}/\MAX(W)$ and at the same time having  $\SUM(W) = O(\MAX(W))$.
In practice, we may want to use a smaller value of $r$:
While $r$ 
does not affect sketch size, it does affect element processing computation.  The basic algorithm is linear in $r$.
In \onlyinproc{the full version}\notinproc{Appendix~\ref{efficientpoint:sec}} we provide a more efficient algorithm with 
computation that is linear in the
number of generated output keys  $O(r (1-\exp(-t \text{\em{e.value}} ))$.
We note that when $\SUM(W) \geq \epsilon^{-2.5} \MAX(W)$ then $r=1$ suffices.

\section{The soft capping span} \label{softspan:sec}

  The {\em soft capping span} $\overline{\softCap}$ contains all
  functions $f$ that can be expressed as nonnegative linear combinations
  of $\softCap_T$  functions.
Equivalently, for some $a(t)\geq 0$,
\begin{equation} \label{idealinverse:eq}
f(w) = \LapM[a](w) = \int_0^\infty  a(t) (1-e^{-wt})dt \ . 
\end{equation}
Note that  $a(t)$ is 
the inverse Laplace$^c$ transform of $f(w)$:
\begin{equation} \label{invtransform}
a(t) = (\LapM)^{-1}[f(w)](t)\ . 
\end{equation}
The following is immediate from the definition \eqref{idealinverse:eq}
\begin{lemma}  \label{invtranscond:lemma}
For any well defined $f(w)$, $\int_{0}^1 a(t) t dt<\infty$ and 
$\int_1^\infty a(t) dt < \infty$. 
\end{lemma}
 The soft capping span includes basic functions such as 
$\softCap_T(w)$ for all $T>0$, $w^p$ for all $p\in (0,1)$, and 
$\ln(1+w)$.  Explicit expressions for the inverse $\LapM$ transforms of these 
functions are provided in Table~\ref{softtran:table}.  The table also 
lists other expressions that we will use for sketching the 
statistics. 
\notinproc{The derivations 
are established in Appendix~\ref{invlapc:sec}.}

\begin{table*}\caption{Inverse $\LapM$ transform of basic functions\label{softtran:table}}
\center 
\begin{tabular}{l|l||l|l}
\hline 
$f(w)$ & $a(t)=(\LapM)^{-1}[f(w)](t)$ & $\int_\tau^\infty a(t) dt$ & $\int_0^\tau a(t) t dt$  \\
\hline 
$\softCap_T(w)$  &  $T \delta_{1/T}(t)$ & $T$, when $y\leq 1/T$; $0$
                                          otherwise & $0$ when $\tau<1/T$\\
$w^p$ ($p\in (0,1)$) & $\frac{p}{\Gamma(1-p)} t^{-(1+p)}$ &
                                                            $\frac{1}{\tau^p \Gamma(1-p)}$ &$ \frac{p}{(1-p)\Gamma(1-p)}\tau^{1-p}$ \\
$\sqrt{w}$ & $\frac{1}{2\sqrt{\pi}} t^{-1.5}$ & $\frac{1}{\sqrt{\pi \tau}}$ & $\frac{\tau^{0.5}}{\sqrt{\pi}}$ \\
$\log(1+w)$ &  $\frac{e^{-t}}{t}$ & $\int_\tau^\infty \frac{e^{-t}}{t} dt
                                    = -\tau\text{Ei}(-\tau)$ &  $1-e^{-\tau}$
\end{tabular}
\end{table*}

We can express $f(W)$ in terms of the inverse $\LapM$ transform $a(t)$ of
$f(w)$ and the transform $\LapM[W]$ of the frequencies:
\begin{eqnarray}
f(W) &=& \int_0^\infty f(w) W(w) dw \nonumber \\
&=& \int_0^\infty W(w) \LapM[a](w) dw \nonumber\\
&=& \int_0^\infty a(t) \int_0^\infty W(w) (1-e^{-wt}) dw dt \nonumber \\
&=& \int_0^\infty a(t) \LapM[W](t) dt \label{fromL:eq}\ . 
\end{eqnarray}
Alternatively, when the inverse transform has a discrete form, that is,
when we can express $f$ using
$\{a_t\}$ for $t\in Y$ as 
\begin{equation} \label{discrete:eq} 
f(w) = \sum_{t\in Y} a_t (1-e^{-wt}) \ ,
\end{equation}
 (Equivalently, $a(t)$ is a 
linear combination of Dirac delta functions at $t\in Y$). 
We can express $f(W)$ in terms of corresponding 
points of $\LapM[W]$:
$$f(W) = \sum_{t\in Y} a_t  \LapM[W](t)\ .$$

We will find it convenient to work with the notation:
\begin{equation} \label{comb:eq}
\LapM[W][a]_\tau^b \equiv \int_\tau^b a(t) \LapM[W](t) 
dt\ . 
\end{equation} 
When the subscript or superscript are omitted, we use the defaults $\tau=0$ and $b=\infty$.
With approximate $\LapM$ measurements we have
\begin{equation} \label{approxLapest:eq}
\widehat{\widehat{\LapM}}[W][a]_\tau^b  = \int_\tau^b a(t) \widehat{\widehat{\LapM}}[W](t) dt\ . 
\end{equation}
When $f(w)$ is such that the inverse $\LapM$ transform $a(t)$ is a
sum of Dirac delta functions,  we can approximate $f(W)$ using
corresponding  point measurements of $\LapM[W]$. In the next section we introduce combination $\LapM$
measurements which allow us to approximate $\LapM[W][a]$  for any
nonnegative continuous $a(t) \geq 0$.

% We will show in Lemma~\ref{comboquality:lemma} and Corollary~\ref{fullrangequality:coro} that 
% when $a(t) \geq 0$, that is, $f\in\overline{\softCap}$,  the estimates we 
% obtain for $\hat{f}(W)$

\section{Combination $\LapM$  measurements} \label{weightedDC:sec}
In this section we show how to sketch $\LapM[W][a]_0^\infty$ for any
$a(t) \geq 0$ that satisfies the conditions in Lemma~\ref{invtranscond:lemma}.
% $a(t) = \LapM^{-1}[f(w)](t)$ where $f\in \overline{\softCap}$ . 
Recall that his allows us to sketch any $f(W)\in \overline{\softCap}$   (see Section~\ref{softspan:sec}).

We define randomized mapping functions of elements, tailored to some
$a(t) \geq 0$ and $\tau>0$, such that the expectation of the (scaled) max-distinct
statistics of output elements is equal to 
$\LapM[W][a]_\tau^\infty$ (see notation \eqref{comb:eq}) we also
establish concentration around the expectation.  We then estimate the component
$\LapM[W][a]_0^\tau$ separately through $\widehat{\SUM}(W)$.

\subsection{Element mapping}
% Consider estimating the $f$-statistics for $f\in
% \overline{\softCap}$.
% First, we let $a(t) \geq 0$ be the inverse $\LapM$ transform of $f$ as
% in \eqref{invtransform}.  Note that by definition, $a(t)$ is defined and
% nonnegative for $f\in \overline{\softCap}$.

 Consider $a(t) \geq 0$. 
  Our element processing is a simple modification of the element processing
  Algorithm~\ref{elemmap1:alg}.  The algorithm inputs the function
  $a()$ (instead of $t$) and  returns output elements (outkey
  and value pairs) instead of only returning outkeys.  
Pseudocode is    provided as Algorithm~\ref{mxdistinct:algo}.
% The value returned for an out element is $\int_{y_i}^\infty a(t) dt$ whenever $y_i$ is 
%  such that the definite integral is positive.   
\begin{algorithm2e}\caption{\protect\Oelems$_{a,H}(e)$: Map of element
    $e$ to a set of output elements for a $\LapM[W][a]_\tau^\infty$ measurement \label{mxdistinct:algo}}
\KwIn{Element $e=(\text{\em{e.key}},\text{\em{e.value}})$, $a(t)\geq 0$,
% ($a(t) \gets   (\LapM)^{-1}[f(w)](t)$), 
integer $r \geq 1$, hash
  functions $H_i$ $i\in[r]$, $\tau>0$}
\KwOut{$\Oelems$: A set of at most $r$ output elements}
\DontPrintSemicolon
$\Oelems \gets []$\;
\ForEach{$i\in [r]$}{$y_i \sim
\Exp[\text{\em{e.value}}]$ \tcp*[f]{independent exponentially distributed with parameter $\text{\em{e.value}}$}\;
$v \gets \int_{\max\{\tau,y_i\}}^\infty a(t) dt$\;
\If{$v>0$}{$\Oelems.append((H_i(\text{\em{e.key}}),v))$ \tcp*[f]{New output element}\;
}
}
\Return{$\Oelems$}
\end{algorithm2e}

 We show that the max-distinct statistics of the output elements divided by $r$ 
\begin{equation}
\widehat{\LapM[W][a]}_\tau^\infty = \frac{1}{r} \MdCount\left(\bigcup_{e\in W}
  \Oelems_{a}(e)  \right)\ 
\end{equation}
has expectation equal to $\LapM[W][a]_\tau^\infty$ with good concentration:
\begin{lemma}  \label{comboquality:lemma}
\begin{equation*}
\E[\widehat{\LapM[W][a]}_\tau^\infty] = \LapM[W][a]_\tau^\infty 
\end{equation*}
\begin{eqnarray*}
 \lefteqn{\forall \delta<1, \,
 \Pr[\frac{|\widehat{\LapM[W][a]}_\tau^\infty - \LapM[W][a]_\tau^\infty}{\LapM[W][a]_\tau^\infty} \geq
\delta]}\\ 
&\leq& 2 \exp(-r \delta^2 \LapM[W][\tau]/3)\ .
\end{eqnarray*}
\end{lemma}

\begin{proof}
The claim on the expectation follows from linearity of expectation and
the claim for point
measurements for each $t$ in Lemma~\ref{unbiasedpoint:lemma}.
The concentration follows from Lemma~\ref{chernoff:lemma} which
establishes
point-wise concentration at each $t$, combined with the
relation
$$\min_{t\in [\tau,\infty)} \LapM[W](t) = \LapM[W](\tau)\ $$
which follows from monotonicity of $\LapM[W](t)$.
\end{proof}
 Note that the assumption $a(t) \geq 0$ is necessary for correctness. 
It ensures monotonicity of $\int_{y}^\infty a(t) dt$ in $y$ which implies that 
the maximum indeed corresponds to minimum $y$.  

  We now address quality of approximation.  From the lemma, quality
is bounded by a function of  $\LapM[W](\tau)$.
Using our treatment of point measurements, we know that it suffices to
ensure that $\tau$  and $r$ are large enough so
that  \eqref{rcond:eq} holds. As with point
measurements, we can use $\tau =
 \sqrt{\epsilon}/\MAX(W))$ and $r$ as in \eqref{rseamless:eq}
to obtain a concentrated measurement of
$\LapM[W][a]_\tau^\infty$.   

 To obtain an estimate of $\LapM[W][a]_0^\infty$, we need to 
separately estimate $\LapM[W][a]_0^\tau$. From
Lemma~\ref{relevantrange:lemma}, when $\tau \leq
\sqrt{\epsilon}/\MAX(W)$, we have
\begin{equation} \label{headest:eq}
\LapM[W][a]_0^\tau \approx \int_0^\tau a(t) t \SUM(W) dt \approx
\widehat{\SUM(W)} \int_0^\tau a(t) t dt\ .
\end{equation}

Closed expressions for $\int_0^\tau a(t)tdt$ (used in \eqref{headest:eq} and for $\int_\tau^\infty
a(t)dt$ (used in Algorithm~\ref{mxdistinct:algo})  for inverse transforms of basic functions are provided in
Table~\ref{softtran:table}. Note that these expressions are bounded and well defined
for all $a$ that are inverse $\LapM$ transform of a $\overline{\softCap}$
function  (see  \onlyinproc{full version}\notinproc{Lemma~\ref{cumbounds:lemma}}).
\ignore{
Note that $\sum_x w_x$ is trivial to compute exactly using a
composable structure that stores its magnitude.  It can also be
approximated unbiasedly with good concentration and NRMSE $\epsilon$ 
by a structure of size $O(\epsilon^{-2}+ |\log\log \sum_x
w_x|)$  using composable weighted Morris
counters~\cite{Morris77,ECohenADS:TKDE2015}.  Therefore, we can obtain
tight estimates of $\LapM[W][a]_0^\tau$ using composable structures of
double logarithmic state.
}

 Our sketch consists of a $\widehat{\SUM}$ sketch applied to data
 elements $W$ and a  $\widehat{\MdCount}$ sketch applied to output
 element $E$
 produced by applying the element mapping
 Algorithm~\ref{mxdistinct:algo} to $W$.
The final estimate we return is
\begin{equation} \label{fest:eq}
\widehat{\widehat{\LapM}}[W][a]_0^\infty =
\frac{1}{r}\widehat{\MdCount}(E) + \widehat{\SUM}(W) \int_0^\tau
a(t)tdt\ .
\end{equation}
There is one subtlety here that was deferred for the sake of exposition:  We do not know $\MAX(W)$ and
therefore can not simply set $\tau$ prior to running the algorithm.
To obtain the worst-case bound we set $r$ as in \eqref{rseamless:eq} and 
set  $\tau$ adaptively to be the
$\ell=3\epsilon^{-2}$ smallest $y$ value generated for a distinct
output key.  To do so, we extend our sketch to include another
component, a ``sidelined'' set, which is the $\ell$ distinct output keys
with smallest $y$ values processed so far. Note that this sketch
component is also composable.
 The sidelined elements are
not immediately processed by the $\MdCount$ sketch:  The
processing is delayed to when they are not sidelined, that is, to the point that due to
merges or new arrivals, there are $\ell$ other distinct output keys with lower
$y$ values.  Finally, when the sketches are used to extract an
estimate of the statistics, we set
$\tau$ as the largest $y$ in the sidelined set, feed all the sidelined keys to the
$\MdCount$ sketch with value $\int_\tau^\infty a(t)dt$, and apply the estimator~\eqref{fest:eq}.
Note that explicitly maintaining the sidelined output keys and $y$ value pairs
requires $\epsilon^{-2}\log n$ storage.  Details on doing so in double logarithmic size  are outlined in 
\notinproc{Section~\ref{sidelined:sec}}\onlyinproc{the full version}.

\subsection{Max-distinct sketches} \label{maxdistinctcount:sec}
Popular composable weighted sampling schemes such as
the with-replacement $k$-mins samples and the
without-replacement  bottom-$k$ samples
\cite{Rosen1972:successive,Rosen1997a,bottomk07:ds,bottomk:VLDB2008} 
naturally extend to data elements with non-unique keys where the weight $m_x$ of the key $x$ is
interpreted as the maximum value of an element with key $x$.
A bottom-$k$ sample over elements with unique keys is
\cite{Rosen1972:successive} 
% priority (sequential Poisson) sample \cite{Ohlsson_SPS:1998,DLT:jacm07}
computed by 
drawing for each element an independent random rank value that is
exponentially distributed with parameter $\text{\em{e.value}}$
$r_e \sim \Exp[\text{\em{e.value}}]$.
The sample includes the $k$ distinct keys with minimum $r_e$ and the
corresponding $m_x$.  It is easy to
see that this sampling scheme can be carried out using composable sketches that stores $k=\epsilon^{-2}$
key value pairs.  For max-distinct statistics we instead use the ranks
$r_e = -\ln u_{\text{\em{e.key}}}/\text{\em{e.value}}$, where $u_{\text{\em{e.key}}} \sim U[0,1]$ is a 
hash function that maps keys to independent uniform random
numbers. These ranks are exponentially distributed, but consistent for
the same key so that the minimum rank of a key is exponentially
distributed with parameter equal to $m_x$.
 We can apply an estimator to this sample
to obtain an estimate of the max-distinct statistics of the data with CV 
$1/\sqrt{k-2}\approx \epsilon$.
% \cite{ECohenADS:TKDE2015}.  
% The sample has size
% $O(\epsilon^{-2} \log n)$ to store a nearly-unique hash of $\text{\em{e.key}}$ and
% $m_x$. 
Through technical use of offsets and randomized rounding and using 
 with-replacement sampling we can design 
max-distinct sketches of size  $O(\log\log n +
\epsilon^{-2}\log\epsilon^{-1})$ (assuming $1\leq m_x = O(\text{poly}(n))$.) 
The design relates to MinHash sketches with base-$b$ ranks
\cite{ECohenADS:TKDE2015}\notinproc{ and is detailed in Section~\ref{mxdistinctcompact:sec}}.

\ignore{
\subsection{Subtleties}

 We now note the technicality that we do not know $\max_x w_x$ in
 advance and thus can not work with a fixed $\tau$.  We  explain how
 to get  around this issue while still retaining our worst-case
 statistical guarantees on quality and bound on the composable structure
 size.   We first note that it suffices to work with $\tau$ large enough
 so that $\LapM[W](\tau) \geq 1$ and small enough so that it is $O(\tau)$.  The measurement is well concentrated,
 so we expect very few keys with minimum $y$ value below $\tau$.  We
 can therefore 
take $\tau$ to be the $\ell$ smallest $y$
 value of a distinct output key, for a small $\ell$.
To do so on the go, we 
separately store  the $\ell$
 output keys with largest values (smallest $y$ values), updating the
 $y$ value as needed when additional elements are processed.
  To do so, it suffices to store hashed
 keys to a small $O(\ell)$ set and their (approximate) minimum $y$
 value.   The contribution of keys that leave the set is added to the
 counter at that point with their full value.  The contribution of the remaining keys is
 factored in the final estimate when $\tau$ is determined.  At that
 point we also add an estimate of $\LapM[W][a]_0^\tau$.
  }

\section{Full range $\LapM$ measurements} \label{allt:sec}

  We consider now computing approximations of the transform $\LapM[W](t)$ for all 
  $t$, which we refer to as a 
  {\em full range measurement}.  This is 
a representation of a function that returns  ``coordinated'' $\LapM[W](t)$
measurements for any $t$ value.  
Our motivation is that an approximate full-range measurement of 
$W$ provides us with estimates of the statistics $f(W)$ for {\em any} $f\in \overline{\softCap}$. 

%  We first consider a smooth approximation of $\widehat{\LapM}[W](t)$
%  obtained by coordinating measurements for different $t$. 
Concretely, consider the set of output keys
$\Okeys_{t,H}(e)$
generated by Algorithm~\ref{elemmap1:alg} for input element $e$
when fixing the parameter $r$, the set of hash functions $\{H_i\}$,
and the randomization $\{y_i\}$,   but varying $t$.  
It is easy to see that the set  $\Okeys_{t}(e)$ monotonically increases with 
$t$  until it reaches size $r$.  
We can now consider all outkeys generated for input $W$ as a function 
of $t$
$$\Okeys_{t}(W) \equiv \bigcup_{e\in W} \Okeys_{t}(e) \ .$$
The number of distinct outkeys 
increases with $t$ until it reaches size $rn$, where $n$ is the number 
of distinct input keys. 
Our full-range measurement is accordingly defined as the function
\begin{equation}
\widehat{\LapM}[W](t) = \frac{1}{r} \bigg|\Okeys_{t}(W) \bigg| \ . 
\end{equation}

 Equivalently, this can be expressed through the 
element mapping provided as Algorithm~\ref{atdistinct:algo}.  For 
each input element $e$, the algorithm generates $r$ output elements 
$\Oelems_{H}(e)$ that are 
outkey and value pairs.  We denote the set 
    of output elements generated for all input elements by $\Oelems(W)$.  We then have 
\begin{equation} \label{allrange}
\widehat{\LapM}[W](t) = \frac{1}{r} \dCount\{e\in \Oelems(W) \mid    \text{\em{e.value}}\leq t \}\ . 
\end{equation}
The number of distinct keys in output elements that have value at most $t$. 
Note that as with point measurements, for the regime of $t$ values where $|\Okeys_t(W)|\leq 
3\epsilon^{-2}$, we use instead $\widehat{\LapM}[W](t) \approx t\SUM(W)$. 

We can also define a combination measurement obtained from a full
range measurement as $$\widehat{\LapM}[W][a] = \int_0^\infty a(t)
\widehat{\LapM}[W](t) dt\ .$$
We show that this formulation is equivalent to directly obtaining a
combination measurement:
\begin{corollary} \label{fullrangequality:coro}
Lemma~\ref{comboquality:lemma} also applies to 
$\widehat{\LapM[W][a]}$ computed from a full-range measurement 
\end{corollary}
\begin{proof}
  Consider element mappings for
full-range and combination measurements that are performed using the same 
randomization (hash functions $H$ and draws $y$).
Consider 
the combination measurement computed from $|\Okeys_t(W)|$. 
 It is easy to verify that 
\begin{eqnarray*}
\widehat{\LapM[W][a]} &=& \frac{1}{r} \MdCount\left(\bigcup_{e\in W}
  \Oelems_{a}(e)  \right) \\ &=& \frac{1}{r} \int_0^\infty a(t) 
|\Okeys_t(W)| dt\ . 
\end{eqnarray*}
\end{proof}

We next discuss composable {\em all-threshold}
distinct counting sketch that when applied to $\Oelems(W)$ allow us
to obtain an approximate 
point measurement of
$\widehat{\widehat{\LapM}}[W](t)$ for any $t$ or compute an approximate combination measurement $\widehat{\widehat{\LapM}}[W][a]$ for any $a \in \LapM^{-1}(\overline{\softCap})$.

\begin{algorithm2e}\caption{\protect\Oelems$_{\{H_i\}}(e)$: Map of 
    element $e$ to a set of output elements for full range measurement \label{atdistinct:algo}}
\KwIn{Element $e=(\text{\em{e.key}},\text{\em{e.value}})$, integer $r \geq 1$, hash 
  functions $\{H_i\}$ $i\in[r]$}
\KwOut{$\Oelems$: A set of $r$ output elements}
\DontPrintSemicolon
$\Oelems \gets []$\\
\ForEach{$i\in [r]$}{$y_i \sim 
\Exp[\text{\em{e.value}}]$ \tcp*[f]{independent exponentially distributed with parameter $\text{\em{e.value}}$}\;
$\Oelems.append((H_i(\text{\em{e.key}}),y_i))$ \tcp*[f]{New output element}\;
}
\Return{$\Oelems$}
\end{algorithm2e}

 An all-threshold sketch of  a set $E$ of key value pairs is a
% and $t>0$ we define the {\em threshold 
%  distinct statistics} as 
%\begin{equation}
% \TdCount_t(E)  =  \dCount\{e\in E \  \mid   \text{\em{e.value}}\leq t \}\ . 
% \end{equation}  
composable summary structure which allows us to 
obtain an approximation $\widehat{\TdCount}_t(E)$ for any
 $t$  (see \eqref{tdistinct:def}.
The design mimics the extension of MinHash sketches
to All-Distance Sketches \cite{ECohen6f,ECohenADS:TKDE2015}, noting that
Hyperloglog sketches  can be viewed as a degenerate form of MinHash sketches.
\notinproc{Details are provided in Section~\ref{ATsketch:sec}.}

\section{The ``hard'' Capping span} \label{hardcappingtransform:sec}
\notinproc{
The capping transform expresses functions as
 nonnegative combinations of capping functions.  The transform
is interesting to us here because it allows us to 
extend approximations of capping statistics to the rich family of 
 functions with a capping transform. 
}
  We define the {\em capping transform}  of a function $f$ as a function 
  $a(t) \geq 0$ and $A_\infty \geq 0$ such that 
$\int_0^\infty 
a(t) dt < \infty$ and $$f(x) = A_\infty x + \int_0^\infty a(t) \Cap_t(x) dt\ .$$

\notinproc{It is useful to allow the transform to include a discrete and 
continuous components, that is, 
have continuous $a(t)$ and a discrete set $\{A_t\}> 0$ for $t\in Y$
such that 
$$f(x) = A_\infty x + \sum_{t\in Y} A_t \Cap_t(x) + \int_0^\infty a(t) \Cap_t(x) 
dx\ ,$$
and require that $A_\infty+ \sum_{t\in Y} A_t + \int_0^\infty 
a(t) dt < \infty$.  For convenience, however we will use a single 
continuous transform $a()$ but allow it to include a linear 
combination of Dirac 
delta functions at points $Y$.  The one exception is the coefficient
$A_\infty < \infty$ of
$\Cap_\infty(x) = x$ which is separate.
}

We denote by $\overline{\Cap}$ the set of functions with a capping
transform.  
% We refer to
% $\overline{\Cap}$ as the {\em nonnegative span} of the capping functions.
We can characterize this set of functions as follows:
\begin{theorem} \label{hardcaptransform:thm}
Let $f:[0,\infty]$ be a nonnegative, continuous, concave, and monotone 
non-decreasing function such that $f(0)=0$ and 
$\partial_+ f(0) < \infty$.  Then $f\in \overline{\Cap}$ with the 
capping transform:
\begin{eqnarray}
 a(x) &=& \left\{
 \begin{array}{ll} - \partial^2 f(x)    & \text{ when }\, \partial^2_- f(x) = \partial^2_+
  f(x) \\  (\partial_- f(x) - \partial_+
 f(x))\delta_x  & \text{otherwise} \end{array}
\right.\\
 A_\infty &=& \partial f(\infty) 
\end{eqnarray}
where $\delta$ is the Dirac delta function and 
$\partial f(\infty) \equiv \lim_{t\rightarrow   \infty} \partial f(t)$. 
\end{theorem}
\notinproc{
\begin{proof}
From monotonicity, $f$ is differentiable almost 
everywhere and the left and right derivatives $\partial_+ f(x)$ and 
$\partial_- f(x)$ are well defined. 
Because $f$ is continuous and monotone, we have 
\begin{equation} \label{diffrel:eq}
f(w) = \int_0^w \partial_+f(t) dt = \int_0^w \partial_- f(t) \equiv 
\int_0^w \partial f(t) dt\ . 
\end{equation}

 From concavity, for all $x\leq y$, $\partial_+f(y) 
\leq \partial_-f(x)$. 
In particular the slope of $\partial f(x)$ is initially 
 (sub)linear and can only decrease as $x$ decreases. 
In particular the limit $\partial f(\infty) \geq 0$ is well defined.

Since $\partial f(x)$ is monotone, it is differentiable and equality 
holds almost 
everywhere: 
$$\partial^2_+f(x)=\partial^2_-f(x) 
 \equiv \partial^2f(x)\ .$$
From concavity,  we have $\partial^2_+f(x), \partial^2_-f(x) \leq 
0$. 

%  We assume $\partial_+ f(0) < \infty$ and use the notation 
% $\lim_{t\rightarrow   \infty} \partial f(t) = \partial f(\infty)$.  

Note that at all $w$ with well defined $\partial f(w)$, we have 
\begin{equation}   \label{arelation}
\partial f(w) = \partial_+f(0) - \int_0^w a(t) dt\ . 
\end{equation}
In particular, when taking the limit as $w\rightarrow \infty$ we get 
\begin{equation} \label{totA:eq}
A_\infty + \int_0^\infty  a(t) dt = \partial_+f(0)\ . 
\end{equation}

{\small
\begin{eqnarray}
\lefteqn{\int_0^\infty a(t) \Cap_t(w) dt}\\
 &=& \int_0^\infty a(t) \min\{t,w\} dt 
                                   \nonumber \\
&=& \int_0^w t a(t) dt + w \int_w ^\infty a(t) dt \label{ct1} \\
&=& w \int_0^w a(t)dt  - \int_0^w (\partial_+f(0) - \partial f(t))dt +
    w \int_w^\infty a(t)dt \label{ct2} \\
&=& w \int_0^\infty a(t) dt  + f(w) - w \partial_+f(0)  \label{ct3}\\
&=& f(w) - A_\infty w \label{ct4}
\end{eqnarray}}
Derivation \eqref{ct1} uses the definition of $\Cap_t$, 
\eqref{ct2} uses integration by parts and \eqref{arelation}, 
\eqref{ct3} uses \eqref{diffrel:eq}, and finally 
\eqref{ct4} uses  \eqref{totA:eq}. 
\end{proof}
} % notinproc
\notinproc{

 The capping transforms of some example functions are derived below
 and   summarized in Table~\ref{captran:table}.
 Note that the transform is a linear operator, so the 
  transform of a linear combination is a corresponding linear 
  combination of the transforms. 
\begin{itemize}
\item 
$f(w)  = \Cap_T(w)$ has the transform $A_\infty=0$ and $a(x) =
\delta_T(x)$. 
\item 
$f(w)  = w = \Cap_\infty(w)$ has the transform $A_\infty=1$, $a(x)=0$
for all $x$. 
\item 
Moments $f(w) = w^p$ for $p\in (0,1)$. 
Here we assume that $w \geq 1$ and for convenient replace the function 
with a linear segment $f(w) = \min\{w, w^p\}$ at the interval $w\in 
[0,1]$. Note that the modified moment function satisfies the 
requirements of Theorem~\eqref{hardcaptransform:thm}. 
We have $\partial f(w) =1$, $\partial^2 f(w) = 0$ for $w < 1$ and 
$\partial f(w) = p w^{p-1}$ and $\partial^2 f(w) = p(p-1) w^{p-2}$
when $w \geq 1$.  Note that $\partial f(\infty) = 0$ and $\partial_+ f(0)=1$. 

We obtain $A_\infty = 0$, $a(x) = p(1-p)x^{p-2}$ for $x>1$ and $a(1) =
(1-p)\delta_1$. 
\item 
Soft capping $f(w)=\softCap_T(w)= T(1-e^{-w/T})$.  We have $\partial f(w) =
e^{-w/T}$ and $\partial^2 f(w) = -\frac{1}{T}e^{-w/T}$. Note that 
$\partial f(\infty) = 0$ and $\partial_+ f(0)=1$. 

We obtain $A_\infty = 0$ and $a(x) = \frac{1}{T}e^{-x/T}$. 
\end{itemize}

\begin{table}\caption{Example capping transforms\label{captran:table}}
\center 
\begin{tabular}{l|l}
\hline 
$f(w)$ & capping transform \\
\hline 
$\Cap_T(w)$  &  $A_\infty=0$, $a(x) =
\delta_T(x)$ \\
$w$ & $A_\infty=1$, $a(x)=0$ \\
$\min\{w, w^p\}$ & $A_\infty = 0$, $a(x) = p(1-p)x^{p-2}$ for $x>1$,\\
&                   $a(x)=0$ for $x<1$,  $a(1) = (1-p)\delta_1$\\
$\softCap_T(w)$ & $A_\infty = 0$, $a(x) = \frac{1}{T}e^{-x/T}$ \\
$\log(1+w)$ & $A_\infty = 0$, $a(x) = \frac{1}{(1+x)^2}$
\end{tabular}
\end{table}

} % not in proc

\subsection{Sketching with signed inverse transforms} \label{useLapM:sec}\label{qualitycomb:sec}

We now consider sketching 
$f$-statistics that are not in 
$\overline{\softCap}$ but where $f$ has a 
%  We outline a recipe to (approximate) 
% $f$-statistics of $W$ from the (approximate) (point, full-range, combination) $\LapM[W]$ measurements. 
{\em signed} inverse transform 
$a(t) = \LapM^{-1}[f(w)](t)$  \eqref{idealinverse:eq}.  
%For functions in 
%$\overline{\softCap}$ there is always nonnegative inverse transform $a()$, but here we 
% also consider signed $a()$. 
We use the notation $a(t) = a_+(t) -
a_-(t)$ where 
\begin{equation} \label{pm:def}
 a_+(t) = \max\{a(t),0\}\,  \text{ and }\,    a_-(t) = \max\{-a(t),0\}
 \ .   
\end{equation}
We define 
$f_+(w)  = \LapM[a_+](w)$ and 
$f_-(w)  = \LapM[a_-](w)$. 
Note that for all $w$, $f(w) = f_+(w) -f_-(w)$ and in particular 
$f(W) = f_+(W)-f_-(W)$. 
Since $a_+, a_- \geq 0$,  we can obtain a good 
approximations $\hat{f}_+(W)$ and $\hat{f}_-(W)$  for each of $f_+(W) = \LapM[W][a_+]$ and 
$f_-(W) = \LapM[W][a_-]$ using 
approximate full-range, two combination, or several point measurements when $a$ is discrete 
and small. We approximate $f(W)$ using the difference $\hat{f}(W) = \hat{f}_+(W)-\hat{f}_-(W)$. 

%   When our approximate counters are unbiased, so are the estimates and 
%   we have $\E[\hat{f}(W)]= f(W)$.  We are also interested, however, in 
%   a small relative error for all $W$. 

 The quality of this estimate depends on a 
 parameter $\rho$:
% We can obtain a (worst-case over $W$) statistical guarantee of a small relative error on the 
% difference  when 
%which is the relative gap between the positive and negative 
%contributions:
\begin{equation}\label{stability:eq}
\rho(a) \equiv  \max_w  \max \frac{\LapM[a_+](w)}{\LapM[a](w)}, \frac{\LapM[a_-](w)}{\LapM[a](w)}
\end{equation}
% We use the notion $$\relerr(f,g) = \max_w |g(w)-f(w)|/f(w)$$
 % to denote the maximum relative error of approximating $f$ with $g$. 

\begin{lemma}
For all $W$,  $f(w) = \LapM[a](w)$, 
{\small 
\begin{eqnarray*}
\lefteqn{\frac{|f(W)-\hat{f}(W)|}{f(W)} \leq }\\ &&\rho(a) 
\left(\frac{|\LapM[W][a_+]-\widehat{\LapM}[W][a_+]|}{\LapM[W][a_+]} \ +
\frac{|\LapM[W][a_-]-\widehat{\LapM}[W][a_-]|}{\LapM[W][a_-]} \
\right) 
\end{eqnarray*}
}
\end{lemma}
\notinproc{
\begin{proof}
We use 
$\hat{f}(W) = \widehat{\LapM}[W][a_+] - \widehat{\LapM}[W][a_-]$ and 
$f(W) = \LapM[W][a_+] - \LapM[W][a_-]$. 
Also note that using the definition of $\rho$, for all $W$,
\begin{eqnarray*}
f(W) &=& \int_0^\infty W(w) \LapM[a](w) dw \\
&\geq&  \int_0^\infty W(w) \frac{\LapM[a]_+(w)}{\rho(a)}  dw \\
&=&  \frac{1}{\rho(a)} \int_0^\infty W(w) 
\LapM[a]_+(w) dw = f_+(W)\ . 
\end{eqnarray*}
Symmetrically, 
for all $W$, $\rho f(W) \geq f_-(W)$. 

Combining it all we obtain 
{\small 
\begin{eqnarray*}
\lefteqn{\frac{|\hat{f}(W) - f(W)|}{f(W)}  =}\\
 && \frac{\widehat{\LapM}[W][a_+]-\LapM[W][a_+]}{f(W)}+
  \frac{\widehat{\LapM}[W][a_-]-\LapM[W][a_-]}{f(W)} \\
&\leq& 
  \rho(a) 
       \frac{|\widehat{\LapM}[W][a_+]-\LapM[W][a_+]|}{\LapM[W][a_+]}+
  \rho(a) \frac{|\widehat{\LapM}[W][a_-]-\LapM[W][a_-]|}{\LapM[W][a_-]}
\end{eqnarray*}
}
\end{proof}} % notinproc 
That is, when the components $f_+(W)$ and $f_-(W)$ are estimated 
within relative error $\epsilon$, then our estimate of $f(W)$ has
relative error at most $\epsilon \rho$.  In particular,
the concentration bound in 
Lemma~\ref{comboquality:lemma} holds with 
$\rho\delta$ replacing $\delta$ and the 
sketch size has $\epsilon \rho$ replacing $\epsilon$. 

 When the exact inverse 
  transform $a$ of $f$ is not defined or has a large $\rho(a)$, we 
  look for an {\em approximate} inverse 
  transform $a$ such that $\rho(a)$ is small  and $f(w) \approx \LapM[a](w)$:
\begin{eqnarray} 
&& \relerr(f, \LapM[a]) = \max_{w>0} \frac{|f(w)-\LapM[a](w)|}{f(w)}
   \leq \epsilon \label{inverse1:eq} %\\
%&& \rho(a) \text{  is "small", say $O(1)$}\nonumber 
\end{eqnarray}
\notinproc{
We will apply this approach to approximate hard capping functions. 
In general, fitting data or functions to sums of exponentials is 
a well-studied problem in applied 
 math and numerical analysis. Methods often have stability issues. 
A method of choice 
is the Varpro (Variable projection) algorithm of Golub and Pereyra 
\cite{varpro_Golub:SiNumAs1973} and improvements \cite{varpro2013}
which has Matlab and SciPy implementation. 
}

\subsection{From $\Cap_1$ to  $\overline{\Cap}$}
We show that from an approximate signed inverse transform of $\Cap_1$ we 
can obtain one with the same quality for any $f\in  \overline{\Cap}$.
\notinproc{
We first consider the special case of $\Cap_T$ functions:
\begin{lemma} \label{1toT:lemma}
Let $\alpha$ be such that $\LapM[\alpha]$ is 
an approximation of 
$\Cap_1$.  Then for all $T$,  $\LapM[\alpha_T]$, where $\alpha_T(x) \equiv \alpha(x/T)$
is a corresponding approximation of $\Cap_T$.  That is,
\begin{eqnarray*}
\relerr(\Cap_T,\LapM[\alpha_T]) &=& \relerr(\Cap_1,\LapM[\alpha]) \\
\rho(\alpha_T) &=& \rho(\alpha)\ . 
\end{eqnarray*}
\end{lemma}
\begin{proof}
The claim on $\relerr$ and $\rho$ follows from the pointwise equality 
\begin{equation}
\frac{|\Cap_T(w) - \LapM[\alpha_T](w)|}{\Cap_T(w)} =
\frac{|\Cap_1(w/T) -\LapM[\alpha](w/T)|}{\Cap_1(w/T)}\ . 
\end{equation}
For $\rho$ observe the correspondence 
\begin{equation}
\frac{\LapM[\alpha_{T+}](w/T)}{\LapM[\alpha_T](w/T)} =
\frac{\LapM[\alpha_+](w)}{\LapM[\alpha](w)}
\end{equation}
and similarly for $\alpha_-$.  The maximum over $w$ of both sides is 
therefore the same. 
\end{proof}
}

\begin{theorem}
Let $f\in  \overline{\Cap}$ and let $a(t)$ be the capping transform of 
$f$.
Let $\alpha(x)$ be such that $\LapM[\alpha](w)$ is an 
approximation of $\Cap_1$. 
Then % $\LapM[c]$, where 
\begin{equation}
c(x)  = \int_0^\infty a(T) \alpha(x/T)  dT\ ,
\end{equation}
is an approximate inverse transform of $f$ that satisfies
\begin{eqnarray}
\rho(c) &\leq & \rho(\alpha) \\
\relerr(f,\LapM[c]) &\leq& \relerr(\Cap_1,\LapM[\alpha]) 
\end{eqnarray}
\end{theorem}
\notinproc{
\begin{proof}
Recall from Lemma~\ref{1toT:lemma}, that for all $T$,
$\Cap_T$ is approximated by $\LapM[\alpha_T]$, where 
$\alpha_T(x)=\alpha(x/T)$). We have 
\begin{eqnarray*}
f(w) &=& \int_0^\infty a(T) \Cap_T(w) dT\\
\hat{f}(w) &=& \int_0^\infty a(T) \LapM[\alpha_T](w) dT 
\end{eqnarray*}
The approximation $\hat{f}$ is obtained by substituting respective 
approximations for the capping functions. 
For all $W$,
\begin{eqnarray*}
\relerr(f(W),\hat{f}(W)) &\leq& \max_w \relerr(\Cap_T,\LapM[\alpha_T])\\
&=& \relerr(\Cap_1,\LapM[\alpha])\ . 
\end{eqnarray*}
The inequality follows from $a\geq 0$. 
The last equality follows from Lemma~\ref{1toT:lemma}.

We now show that $\hat{f}(W) = \LapM[c]$:
\begin{eqnarray*}
\hat{f}(w) 
     &=&
\int_0^\infty a(T) \LapM[\alpha_T](w) dT \\
&=& \int_0^\infty \int_0^\infty a(T) \alpha_T(x) (1-e^{-wx})dxdT \\
&=& \int_0^\infty \int_0^\infty a(T) \alpha(x/T)  dT (1-e^{-wx})dx\\
&=& \int_0^\infty c(x) (1-e^{-wx})dx = \LapM[c]\ ,
\end{eqnarray*}
where 
$$c(x)  = \int_0^\infty a(T) 
\alpha_T(x)  dT = \int_0^\infty a(T) 
\alpha(x/T)  dT\ .$$

We will now  show that 
$\forall w,\ \LapM[c_+](w) \leq \rho(\alpha) \LapM[c](w) $.  The claim for $c_-$ is 
symmetric and together using the definition of $\rho$ they imply that $\rho(c) \leq \rho(\alpha)$. 
We first observe that 
\begin{eqnarray*}
c_+(x) &=& \max\{0,\int_0^\infty a(T)  (\alpha_{T+}(x)-\alpha_{T-}(x) 
dT\} \\
 &=&  \max\{0,\int_0^\infty a(T)  (\alpha_{T+}(x)dT-\int_0^\infty a(T)\alpha_{T-}(x)  dT\}\\
&\leq& \int_0^\infty a(T)  \alpha_{T+}(x)  dT\ . 
\end{eqnarray*}
The last equality follows from the two integrals being nonnegative. 

Therefore,
\begin{eqnarray*}
\LapM[c_+](w)  &\leq& \int_0^\infty a(T) \LapM[\alpha_{T+}](w) dT\\
&\leq& \int_0^\infty a(T) \rho(\alpha_T) \LapM[\alpha_{T}](w) dT\\
&=& \rho(\alpha) \int_0^\infty a(T)  \LapM[\alpha_{T}](w) dT\\
&=& \rho(\alpha) \LapM[c](w) \ . 
\end{eqnarray*}
Using Lemma~\ref{1toT:lemma} that established $\rho(\alpha) = \rho(\alpha_T)$. 

\end{proof}
}
 It follows that to approximate $f(W)$ we do as follows. We compute the capping
 transform of $f$ and use an approximate inverse transform of
 $\Cap_1$, from which we obtain an approximate inverse transform $c$
 of $f$.  We then perform 
 two combination measurements (with respect to the negative and 
positive components $c_+$ and $c_-$).  Alternatively, we can use one 
full-range measurement to estimate both.

\ignore{
{\em Generalized Laplace measurements} $\LapM[W](\tau)$, specified by a monotone 
non-increasing  function $\tau(y) \geq 0$ of our choice:
  $$\LapM[W](\tau) \equiv \int_0^\infty  \tau(y)  \int_0^\infty  W(w) w 
  \exp(-wy) dw dy\ .$$

For $\alpha() \geq 0$ of our choice, by defining 
(roughly...) $\tau(x) = \alpha'(x)$ we have 
$$\LapM[W](\tau) = \int_0^\infty \alpha(t)dt - \int_{0}^{\infty} \alpha(t) \Laplace[W](t) dt\ .$$ 
\end{itemize}
Laplace measurements with threshold $t$ can be expressed as 
  generalized Laplace measurements with a 
step function $\tau$, where 
  $\tau(y) =1$ for $y\leq t$ and $\tau(y)=0$ otherwise:
\begin{eqnarray*}
\LapM[W](\tau)  &=& 
\int_0^t  \int_0^\infty  W(w) w \exp(-wy) dw dy \\
&=& \int_0^\infty W(w)   \int_0^t  w \exp(-wy) dy dw \\
&=& \int_0^\infty W(w)  (1-\exp(-wt)) dw \\
&=& 1 - \int_0^\infty W(w) \exp(-wt) = 1-\Laplace[W(w)](t)\ . 
\end{eqnarray*}

}% ignore 

\subsection{Sketching $\Cap_1$} \label{hardcap:sec}
We consider approximate inverse transforms of  $\Cap_1$.
\ignore{
 We consider now approximation of $\Cap_T$ statistics from $\LapM$
 measurements. To do so, as discussed in
 Section~\ref{qualitycomb:sec}, we fit $\Cap_T$ 
to  a linear combination 
$\alpha(t)$ of 
functions of the form $1-\exp(- w t)$.   That is,
$$\LapM[\alpha](w)  = \int_0^\infty \alpha(t) (1-\exp(- t 
w))  dt\ .$$
We would like the relative error of the approximation
$$\relerr(\Cap_T, \LapM[\alpha]) = \max_w \frac{|\Cap_T(w) - \LapM[\alpha](w)|}{\Cap_T(w)}$$
to be small, and also to have well behaved coefficients, that is, have
a small $\rho(\alpha)$ \eqref{stability:eq}.

}
 The simplest approximation (see \eqref{softcap:eq}) is to 
approximate $\Cap_T$ statistics by
$\softCap_T$.  The worst-case  error of this approximation is
$\relerr(\Cap_T,\softCap_T) = 1/e \approx 0.37$.  Note, however,
that the
relative error is maximized at $w = T$, but vanishes when $w \ll T$ and $w \gg T$.  This
means that only distributions that are heavy with keys of weight
approximately $T$ would have significant error.
Noting that $\softCap_T$ is an underestimate of $\Cap_T$, we can
decrease the worst-case relative error using the approximation
\begin{equation} \label{scaledsoft:eq}
\frac{2e}{2e-1} \softCap_T(w)\ 
\end{equation}  
and obtain
$\relerr(\Cap_T, \frac{2e}{2e-1} \softCap_T) = \frac{1}{2e-1} \approx
0.23\ .$
This improvement, however, comes at the cost of spreading the error,
that otherwise dissipated for very large and very small  frequencies $w$,
across all frequencies\notinproc{ (see Figure~\ref{capapproxplot:fig})}.

We derive tighter approximations of  $\Cap_1$ using a signed
approximate inverse transform $\alpha()$.
% we will need
%multiple $\LapM$ point measurements, two combination measurements, or
%a full-range measurement.
% We first argue that without loss of generality it suffices to 
% consider $\Cap_1 = \min\{1,w\}$.  
 We first specify properties of $\alpha()$ so that
 $\LapM[\alpha]$ has desirable properties as an approximation of $\Cap_1$.
To have the error vanish for $w \gg 1$, that is, have $\LapM[\alpha](w)
\rightarrow 1$ when $w\rightarrow +\infty$, we must have 
\begin{equation} \label{largecond:eq}
\int_0^\infty \alpha(t) dt = 1\ .
\end{equation}
To have the error vanish for $w \ll 1$, that is, have $\LapM[\alpha](w)/w
\rightarrow 1$ when $w\rightarrow 0$, we must have 
\begin{equation} \label{smallcond:eq}
\int_0^\infty t \alpha(t) dt = 1\ .
\end{equation}
\onlyinproc{
We show (see full version) that using $\alpha$ of the form
$$\alpha(t) =
(A+1)\delta_1(t) - \alpha_1 \delta_{\beta_1}(t) - \alpha_2
\delta_{\beta_2}(t)\ ,$$ where $\delta$ is the Dirac delta function,
we obtain the following:
{\center
\begin{tabular}{rrr || l | r}
% \multicolumn{3}{parameters} & $\relerr$ & $\rho(\alpha)$ \\
$A$ & $\beta_1$ & $\beta_2$  & $\relerr\approx$ & $\rho(\alpha)<$ \\
\hline
$10.0$ & $0.9$ & $3.75$ & $0.115$ & $12.4$ \\
$1.5$ & $0.6$ & $7.97$ &  $0.14$ & $2.9$ 
\end{tabular}\\
}
}
\notinproc{
We can relate the approximation quality obtained by
$\alpha()$ to its ``approximability'':
\begin{theorem} \label{boundrho:thm}
Let $\alpha()$ be such that \eqref{largecond:eq} and
\eqref{smallcond:eq} hold.  Let $\alpha_+()$ and $\alpha_-()$ be
defined as in \eqref{pm:def}.
Then 
\begin{equation}
\rho(\alpha) \leq \frac{\int_0^\infty \alpha_+(t)
  dt}{1-\relerr(\LapM[\alpha],\Cap_1)}\ .
\end{equation}
\end{theorem}
\begin{proof}
From conditions \eqref{largecond:eq} and \eqref{smallcond:eq} we
obtain
\begin{eqnarray*}
\int_0^\infty \alpha_+(t)dt &=& 1+ \int_0^\infty \alpha_-(t)dt \\
\int_0^\infty t\alpha_+(t)dt &=& 1+ \int_0^\infty t \alpha_-(t)dt 
\end{eqnarray*}

We also have 
\begin{eqnarray}
\LapM[\alpha_+](w) &\leq& \int_0^\infty \alpha_+(t)dt \min\{1,w\} \\
\LapM[\alpha_-](w) &\leq& \int_0^\infty \alpha_-(t)dt \min\{1,w\} \\
\LapM[\alpha](w) &\geq& (1-\relerr(\LapM[\alpha],\Cap_1)) \min\{1,w\}
\end{eqnarray}
The last inequality follows from the definition of $\relerr$.

Combining, we obtain 
\begin{eqnarray*}
\max_w \frac{\LapM[\alpha_+](w)}{\LapM[\alpha](w)} &\leq& 
\frac{\int_0^\infty \alpha_+(t)dt}{1-\relerr(\LapM[\alpha],\Cap_1)}\\
\max_w \frac{\LapM[\alpha_-](w)}{\LapM[\alpha](w)} &\leq& 
\frac{\int_0^\infty \alpha_-(t)dt}{1-\relerr(\LapM[\alpha],\Cap_1)} \\
&=& \frac{\int_0^\infty
  \alpha_+(t)dt - 1}{1-\relerr(\LapM[\alpha],\Cap_1)}  \\
&<& \frac{\int_0^\infty
       \alpha_+(t)dt}{1-\relerr(\LapM[\alpha],\Cap_1)}
\end{eqnarray*}
\end{proof}

 We now consider using simple combinations of three functions of the
 particular form:
$$(A+1)  (1-\exp(-w))- \alpha_1 (1-\exp(-\beta_1 w)) -
\alpha_2 (1-\exp(-\beta_2 w))\ $$
where  $\beta_1< 1< \beta_2$.
Equivalently, we  estimate $\Cap_1(w)=\min\{1,w\}$   using $\LapM[\alpha]$ where $$\alpha(t) =
(A+1)\delta_1(t) - \alpha_1 \delta_{\beta_1}(t) - \alpha_2
\delta_{\beta_2}(t)\ ,$$ where $\delta$ is the Dirac delta function.

From conditions \eqref{largecond:eq} and \eqref{smallcond:eq} we obtain
\begin{eqnarray*}
A  - \alpha_1 -  \alpha_2 &=&  0 \\
A - \alpha_1\beta_1 - \alpha_2\beta_2 & =& 0\ .
\end{eqnarray*}
Therefore we obtain
\begin{eqnarray*}
\alpha_1 = A\frac{(\beta_2-1)}{\beta_2-\beta_1} \\
\alpha_2 = A\frac{(1-\beta_1)}{\beta_2-\beta_1}\ .
\end{eqnarray*}
We did a simple grid search with local optimization over the free parameters
$A,\beta_1,\beta_2$ focusing on small values of $A$.
 Note that $\int_0^\infty \alpha_+(t)dt = A+1$.  Applying 
 Theorem~\ref{boundrho:thm}
we obtain that $\rho(\alpha)\leq (A+1)/(1-\relerr)$, which is a small 
constant when $A$ and $\relerr$ are small. 
Two choices and their tradeoffs are listed in the table below.  Note
the small error (less than $12\%$).  In both cases the error vanishes for 
small and large $w$.   
 We obtain that 
three point measurements with a small relative 
error for the points $\LapM[W](t)$, $\LapM[W](\beta_1 t)$,
$\LapM[W](\beta_2 t)$ yield an approximation of the linear combination 
with quality the same order.\\
{\center
\begin{tabular}{rrr || l | r}
% \multicolumn{3}{parameters} & $\relerr$ & $\rho(\alpha)$ \\
$A$ & $\beta_1$ & $\beta_2$  & $\relerr\approx$ & $\rho(\alpha)<$ \\
\hline
$10.0$ & $0.9$ & $3.75$ & $0.115$ & $12.4$ \\
$1.5$ & $0.6$ & $7.97$ &  $0.14$ & $2.9$ 
\end{tabular}\\
}

} % notinproc
Figure~\ref{capapproxplot:fig}
shows $\Cap_1(w)$ and various
approximations.  The single point measurement approximations:  $\softCap_1$
and scaled $\softCap_1$ and the two 3-point approximations from the
table.  One plot shows the functions and the other shows their ratio to $\Cap_1$ which
corresponds to the relative error as a function of $w$.  The graphs
show that the error vanishes for small and large values of $w$ for all but the
scaled $\softCap_1$ \eqref{scaledsoft:eq}.  We can also see the
smaller error for the 3-point approximations.

\begin{figure*}
\center
\includegraphics[width=0.42\textwidth]{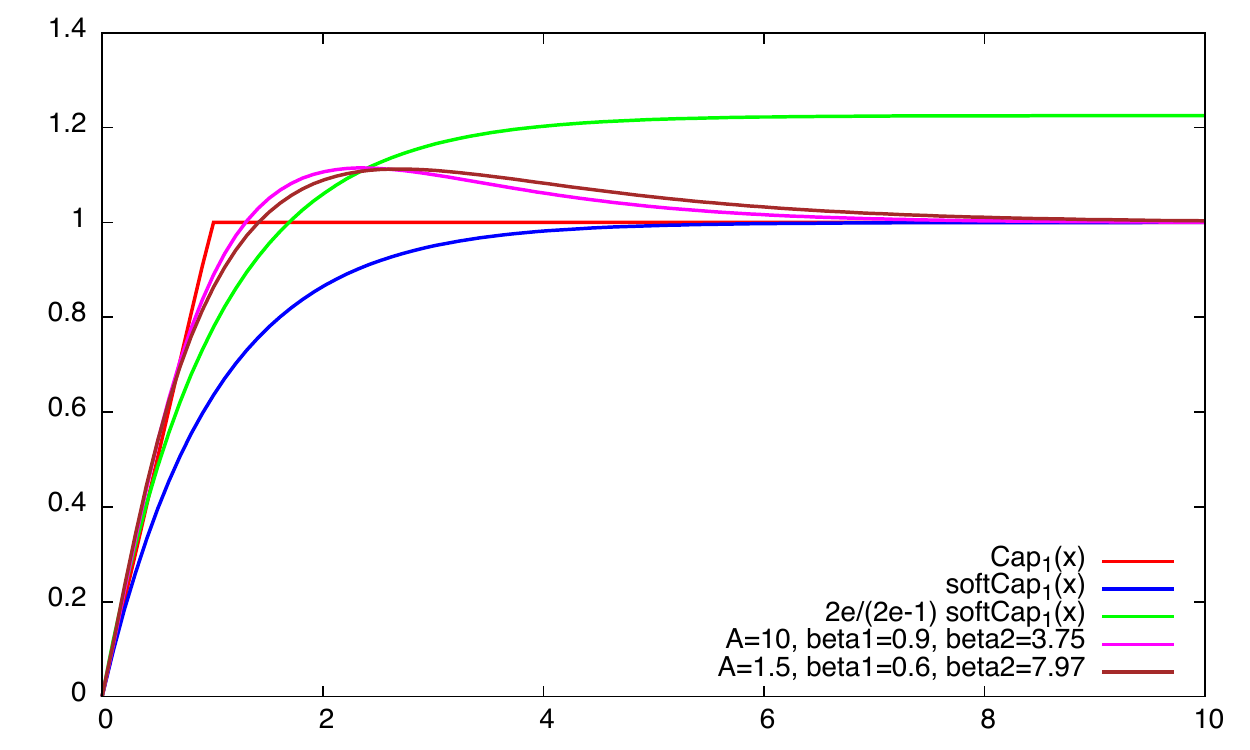}
\includegraphics[width=0.42\textwidth]{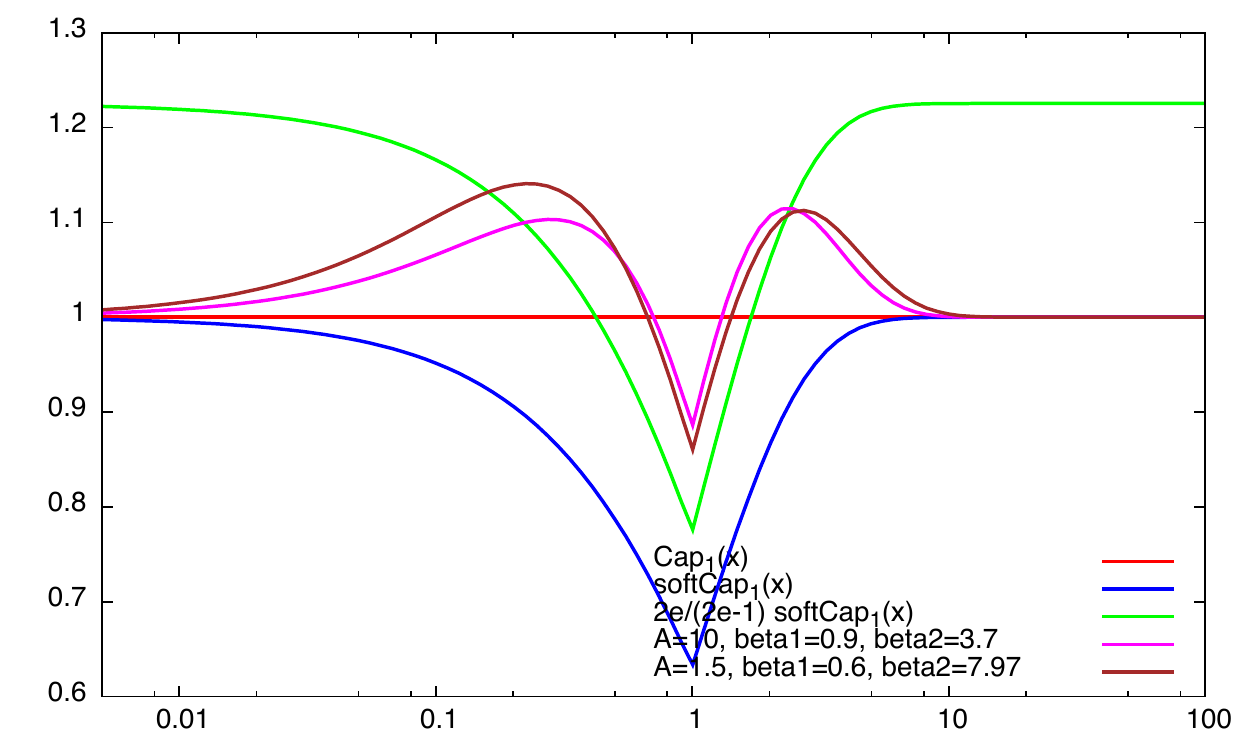}
\caption{Left: $\Cap_1(w)$ function and different approximations:
With a single $\LapM$ point measurement we can use
the soft capping function 
 $\softCap_1(w)$ or scaling it $\frac{2e}{2e-1} \softCap_1(w)$ to minimize
the worst-case relative error.  We also show two 3 point measurements
 approximations with different parameters. Right: The corresponding ratio
to $\Cap_1$ which shows the relative error of the different approximations.\label{capapproxplot:fig}}
\end{figure*}

\section{Experiments} \label{experiments:sec}

  We performed experiments aimed at demonstrating the 
ease of implementing our schemes and explaining the use of
the parameters that control the approximation quality.   We implemented the 
element mapping functions in Python.  We also use Python implementations of 
approximate  $\MxDCount$ and $\dCount$ sketches.
We generated synthetic data using a Zipf distribution of keys. We set
the values of all elements to $1$.
Zipf distributions are parametrized by $\alpha$ that controls
the skew: The distribution is more skewed (has
more higherer frequency keys) with higher $\alpha$.     Typically
values in the range $\alpha\in [1,2]$ model real-world 
distributions.  We used $\alpha=\{1.1,1.2,1.5,2.0\}$ in our experiments.

\subsection{Point measurements}
In this set of experiments we generated data sets with $10^5$
elements and performed point measurements, which
approximate the statistics $\softCap_T(W)$.   We used
$T=\{1,5,20,100,500\}$.  We applied element
mapping Algorithm~\ref{elemmap1:alg} with parameter $r\in\{1,10,100\}$ to generate output elements $E$.  The output
elements were processed by an approximate distinct counter with
parameter $k=100$ to obtain an estimate of $\dCount(E)$.  The final
estimate \eqref{approxpointm:eq} is the approximate count divided by
$r$.  In our evaluation we also computed the exact count
$\dCount(E)$, to obtain the exact point measurement, and also the
exact value of the statistics $\softCap_T(W)$.

 Recall that there are two sources of error in the approximate measurement.  The first is the error
of the measurement $\widehat{\LapM}[W](t)$ ($\dCount(E)/r$) as an estimate of $\LapM[W](t)$.
 This error is controlled by the parameter $r$ in
   the element mapping Algorithm~\ref{elemmap1:alg}.  We showed that
$r>\epsilon^{-2.5}$ always suffices in the worst-case but generally on
larger data set that are not dominated by few keys we have
$\LapM[W](t) > \epsilon^{-2}$ and $r=1$ suffices.
The second source of error is the approximate counter which returns an
approximation $\widehat{\widehat{\LapM}}[W](t)$ of the measurement $\widehat{\LapM}[W](t)$.  An
approximate counter with $k=\epsilon^{-2}$ has NRMSE $\epsilon$.
In the experiments we used a fixed $k=100$ which has $\epsilon=0.1$.
We performed 200 repetitions of each experiment and computed the
average errors.  Results are
illustrated in Figure~\ref{pointexperiments:fig}: The left plot in
the figure shows the average number of distinct output elements
generated for different Zipf parameters $\alpha$ and $T$ value for
$r=1$.  The expectation of $\dCount(E)$ is equal to
$r \widehat{\LapM}[W](1/T)$.
 The middle plot shows the
normalized root mean squared (NRMSE) error of the measurement
$\frac{T}{r}\dCount(E)$ that uses the exact distinct count of output keys.  We can see
that the error rapidly decreases with the parameter $r$ and that it
is very small.  The right plot shows the error of the approximate
measurement as a function of $r$, obtained by applying an approximate
counter to $E$.  The additional error component is an 
error with NRMSE $10\%$ with respect to the measurement, and
is independent  of the measurement error.  We can see that the
approximation error dominates the total error which is about $10\%$.

\begin{figure*}
\center
\includegraphics[width=0.30\textwidth]{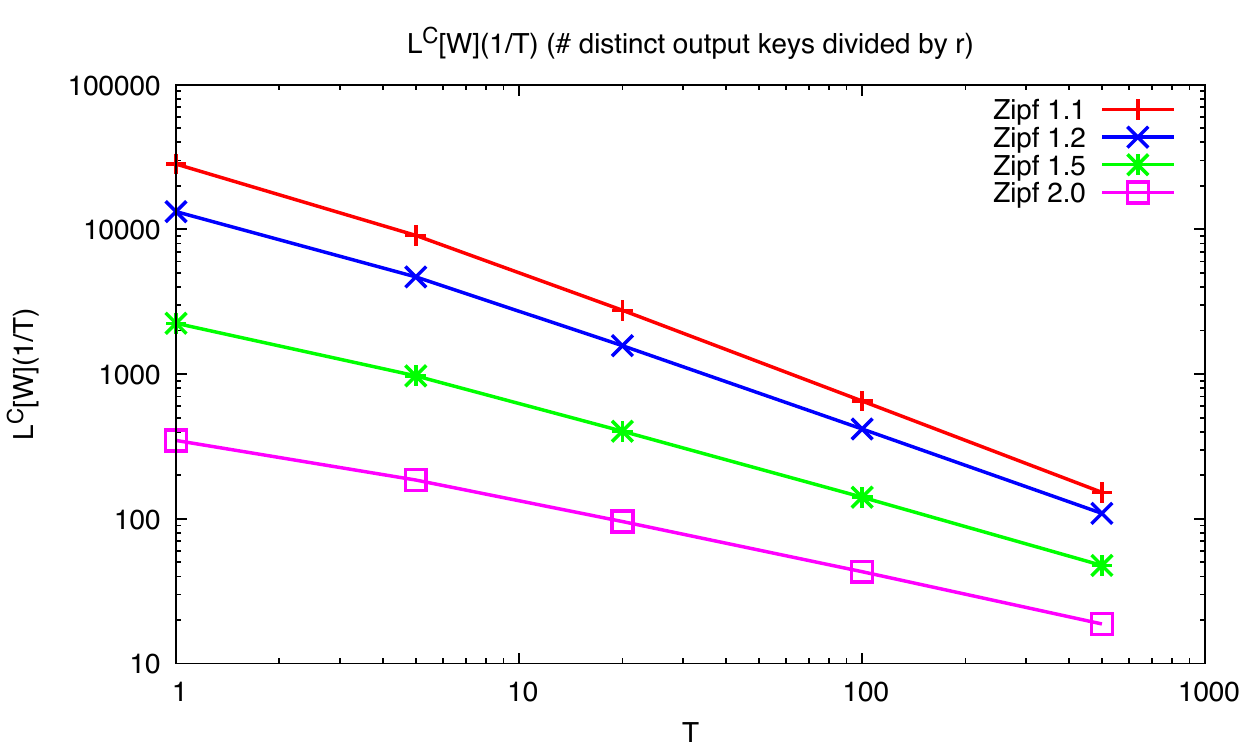}
\includegraphics[width=0.30\textwidth]{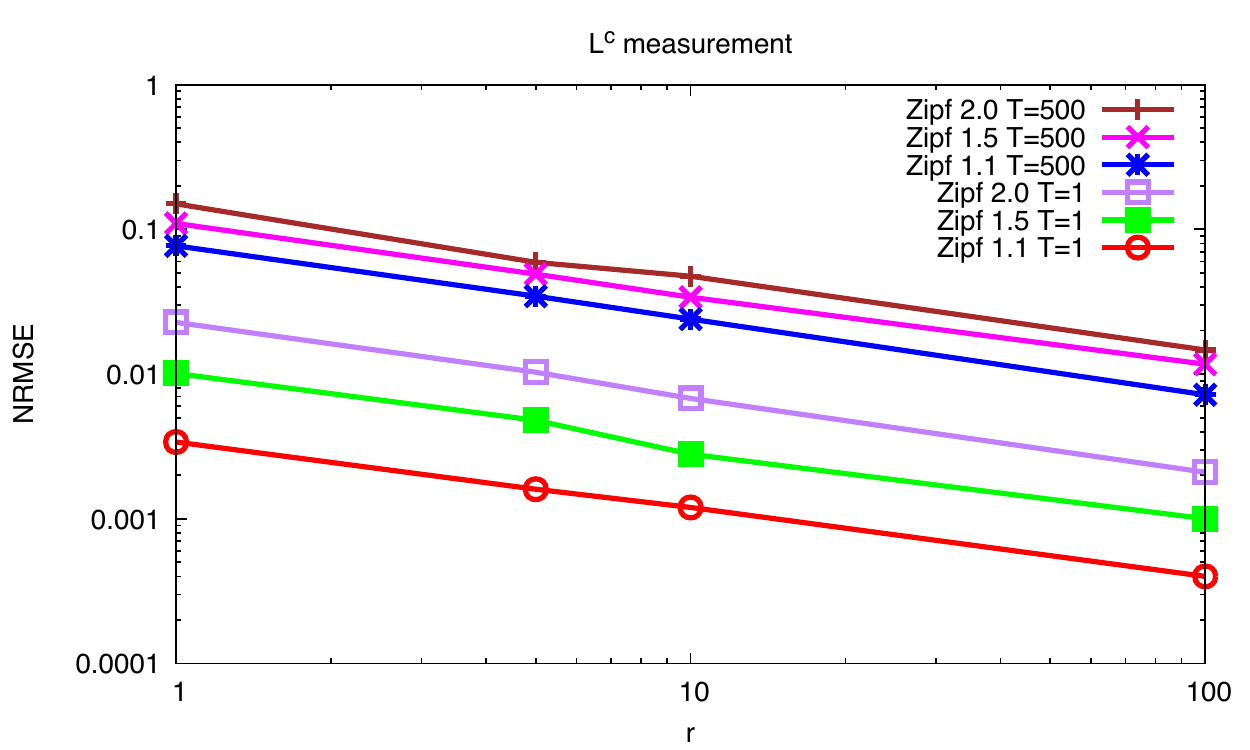}
\includegraphics[width=0.30\textwidth]{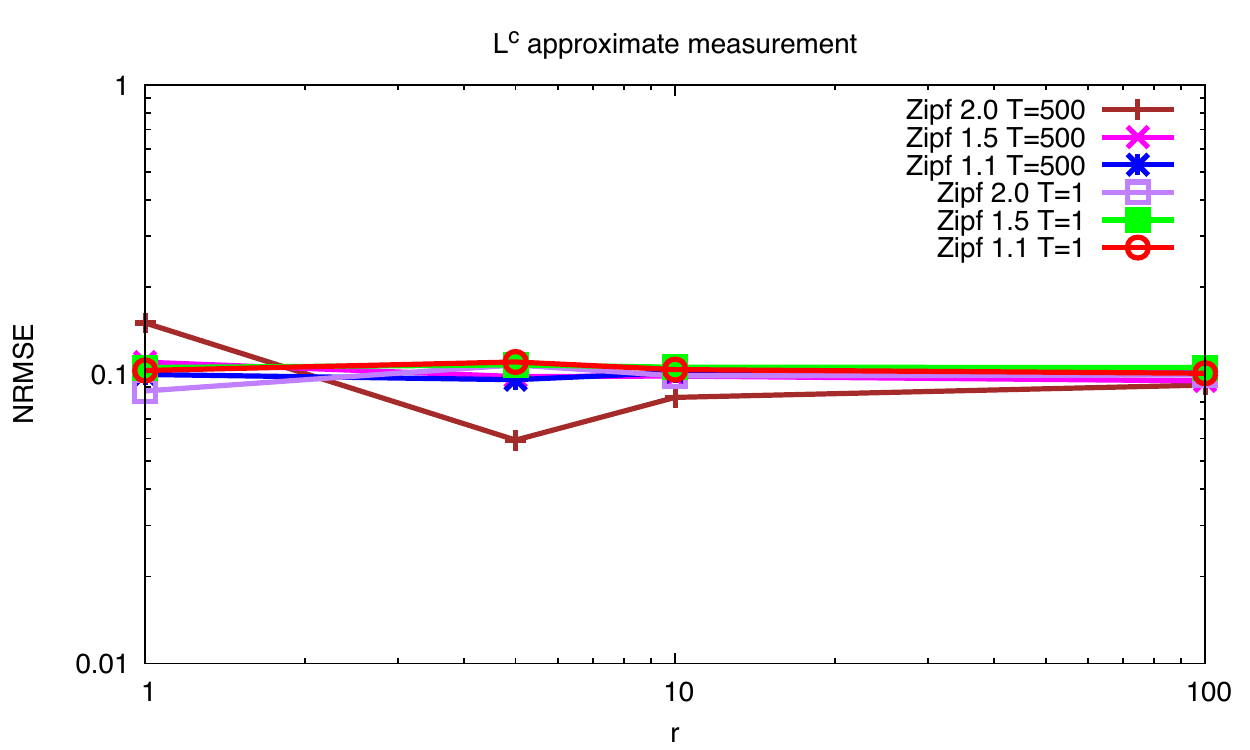}
\caption{Point measurements experiments on $10^5$ elements with Zipf
  keys (averaged over 200 repetitions).  Left: $\LapM[W](1/T)$,
  Middle:  NRMSE of exact measurement $\widehat{\LapM}[W](1/T)$. Right:
  NRMSE of approximate measurement $\widehat{\widehat{\LapM}}[W](1/T)$.\label{pointexperiments:fig}}
\end{figure*}

\subsection{Combination measurements}
In this set of experiments we estimate the statistics $f(W)$ for $f(w) =\sqrt{w}$ using approximate combinations measurements.
% Recall (see Table~\ref{softtran:table}) that the expressions for the integrals of the inverse transform $a(t)$ are
% $A(y) = \int_y^\infty a(t)dt = (\pi y)^{-0.5}$ and
% $C(y) = \int_0^y a(t)t dt = (y/\pi)^{0.5}$.
We used the element mapping Algorithm~\ref{mxdistinct:algo} with
$\tau=0$.  
Note that the estimates we get are unbiased but the NRMSE error can be larger due to the contributions of the $t$ regime with a very small number of output keys.
Figure~\ref{combexperiments:fig} shows the NRMSE of the measurement  $\frac{1}{r}\MdCount(E)$ 
for different values of $r$.   The application of an approximate
$\MdCount$ counter of size $\epsilon^{-2}$ introduces NRMSE  of at most $\epsilon$ to this measurement.
\begin{figure}
\center
\includegraphics[width=0.30\textwidth]{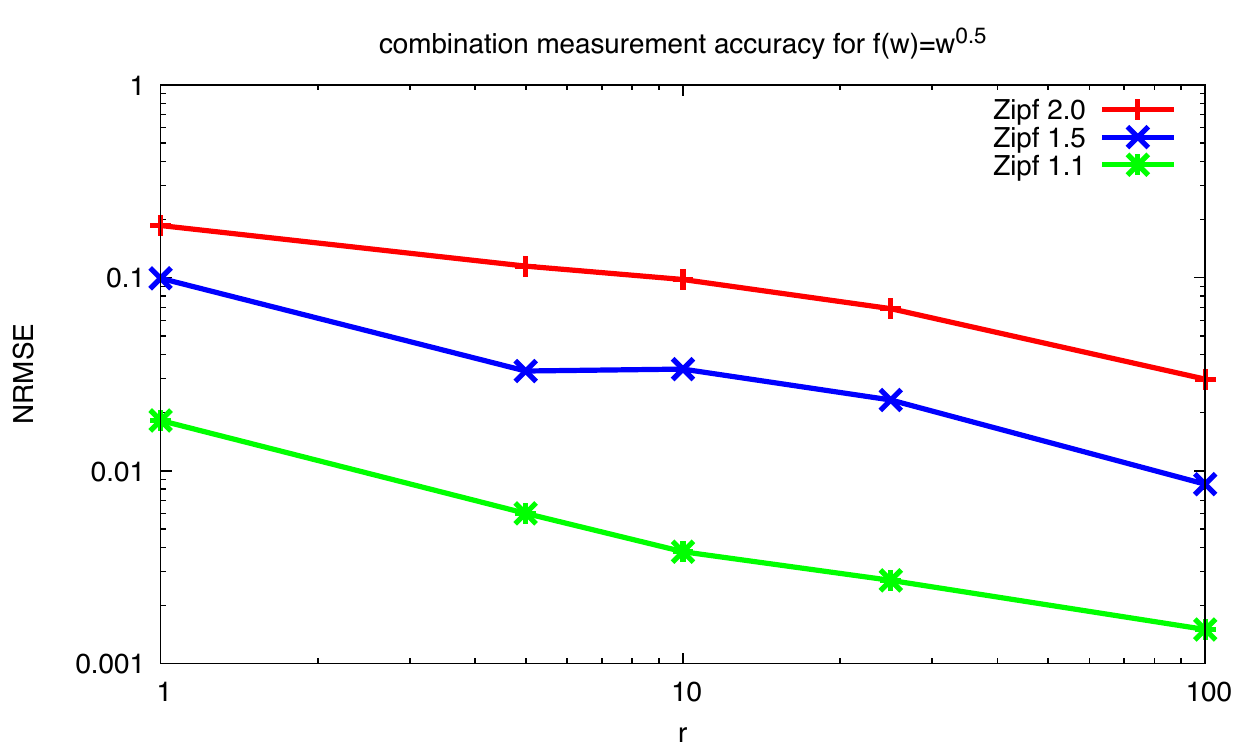}
\caption{Combination measurement NRMSE for $f(w)=\sqrt{w}$ as a
  function of $r$ on $10^5$ elements with Zipf distributed   keys, averaged over 200 repetitions.\label{combexperiments:fig}}
\end{figure}

\section{Conclusion} \label{conclu:sec}
We presented a novel elegant framework for composable
sketching of concave sublinear statistics.  We obtain
state-of-the-art asymptotic bounds on sketch size together with highly
efficient practical solution.  

We leave open the intriguing question of fully understanding the limits of 
our approach and in particular, understand  if  the scope of
sample-based sketching is limited to (sub)linear statistics (we suspect it does). 
Another question is to precisely quantify the approximation tradeoff for $\Cap_1$ (and
hence for any $\overline{\Cap}$ function)

  \small
\bibliographystyle{plain}
\bibliography{cycle}

\onlyinproc{\end{document}}
\appendix
\section{Efficient element processing for point measurements} \label{efficientpoint:sec}
 Ideally, we would want to determine the selected indices
  using computation proportional to $O(r (1-\exp(-t \text{\em{e.value}} ))$, which
  is the expected number of output keys returned.
There are many ways to do so and retain the confidence bounds.  We
recommend a way that uses
{\em varopt dependent sampling}
\cite{Cha82,GandhiKPS:jacm06,CCD:VLDB2011,varopt_full:CDKLT10}.
  This shows that, in principle, we can always obtain very tight concentration for 
  $\widehat{\LapM}[W](t)$ as an approximation of $\LapM[W](t)$  (using large 
enough $r$).    
% We shall see that the quality of our approximate count
% measurement will be bottlenecked by that of the approximate distinct counter.

Pseudocode for the inner loop of the efficient element processing 
is provided as Algorithm \ref{varoptproc:alg}.  The inner loop selects 
the samples from $[1,r]$, each in probability $p$, in time proportional to the number of 
selected samples.  The range 
$[1,r]$ 
is logically partitioned to $\lceil r p \rceil$ consecutive 
ranges of size $\lfloor 1/p \rfloor$, where $p = 1-\exp(-t \text{\em{e.value}})$
(the last range may be smaller).  
The probability associated with 
each range is its size times $p$, which is $p \lfloor 1/p \rfloor \leq 1$.   We then varopt sample 
the $O(pr)$ ranges according to these probabilities (e.g., using $\lceil r p \rceil$
{\em pair aggregations}  \cite{CCD:VLDB2011}). 
Finally, for each range that is included in the sample, we 
uniformly return one of the numbers in the range. 

The improved scheme select a varopt sample of the set $[r]$ such that each $i\in [r]$
is included with probability $1-\exp(-t \text{\em{e.value}})$.  Since the joint 
inclusion/exclusion probabilities have the varopt property, the 
concentration bounds  (Lemma \ref{chernoff:lemma}) hold.  Moreover, generally varopt improves quality 
over Poisson sampling by eliminating the variance in the sample size. 

\begin{algorithm2e}\caption{Element processing inner loop with varopt
    sampling (for $\widehat{\LapM}$
    measurements) \label{varoptproc:alg}}
\DontPrintSemicolon
$p \gets 1-\exp(-t \text{\em{e.value}})$ \tcp*[f]{inclusion probability for each
  $i\in [r]$}
\tcp{Compute vector $\boldmath{\pi}$  of sampling probabilities for 
  consecutive partition of $[1,r]$ to subranges of equal size $\lfloor 
  1/p \rfloor$ and the remainder as the last entry.}
$\ell \gets \lceil \frac{ pr }{p\lfloor 1/p \rfloor} \rceil$\;
$rem \gets r \,\text{mod}\, p\lfloor 1/p \rfloor$ \;

\eIf(\tcp*[h]{$\ell_f$ number of full-size ranges}){$rem = 0$}{$\ell_f = \ell$}{$\ell_f = \ell-1$}\;
Initialize a vector $\boldmath{\pi}$ of size $\ell$
\ForEach{$j \leq \ell_f$}{$\pi_j \gets p \lfloor 1/p
\rfloor$}
\If{$rem \not= 0$}{$\pi_\ell \gets p rem$}
\tcp{varopt sample the entries of $\boldmath{\pi}$ }
$S\gets $ varopt sample $\boldmath{\pi}$ \tcp*{$S$ has  $\lfloor pr
  \rfloor$ or $\lceil pr \rceil$ entries equal to $1$ and remaining entries
  are $0$.}
\ForEach{$h =1,\ldots,\ell_f$}{
\If{$S_h = 1$}{$j \sim U[1,\lfloor 1/p \rfloor]$ \tcp*{uniformly at
  random}\; $i \gets (h-1)\lfloor 1/p \rfloor + j$\;
  $\Okeys.append(H_i(\text{\em{e.key}})) $ }
}
\If{$rem \not= 0$ and $S_\ell =1$}{$i \sim U[r-rem+1,r]$ \tcp*{uniform at random}\;
$\Okeys.append(H_i(\text{\em{e.key}})) $}
\end{algorithm2e}

We point out a choice of $r$, for
uniform element values ($\text{\em{e.value}} = 1$ for all 
  elements), which {\em seamlessly} guarantees small element processing and
  tight confidence bounds also when $t$ is very small:
\begin{corollary}.
With uniform elements and $r = c\epsilon^{-2}/t$, we have 
$O(c\epsilon^{-2})$ element processing and probability of relative
error that exceeds $\epsilon$ that is bounded by $2 \exp(-c /5)$.
\end{corollary}
\begin{proof}
The element processing is proportional to the number of outkeys
computed and is $\frac{r}{t} (1-\exp(-t))  \leq c \epsilon^{-2}$.
For the concentration bound we note that $w_x \geq 1$ for all active keys, from \eqref{trivbound:eq}, $\LapM[W](t) \geq (1-1/e) \min\{1, t\}$.
\end{proof}

\section{Inverse $\LapM$ transform derivations} \label{invlapc:sec}
In this section we derive the expressions for the inverse  $\LapM$ transform in
Table~\ref{softtran:table}.  Our main tool is the following Lemma
which expresses
the inverse $\LapM$ transform of $f$ in terms of the inverse Laplace
transform of the derivative of $f(w)$:
\begin{lemma} \label{invlap:eq}
$$ (\LapM)^{-1}[f(w)](t) = \frac{1}{t}\Laplace^{-1}[\frac{\partial f(w)}{\partial w}](t)\ ,$$
where $\Laplace$ is 
the Laplace transform. 
\end{lemma}
\begin{proof}
We look for a solution $a(t)$ of 
$$f(w) = \int_0^\infty a(t) (1-e^{-wt}) dt\ . $$
Differentiating both sides by $w$ we obtain 
$$\frac{\partial f(w)}{\partial w} = \int_0^\infty t a(t) 
e^{-wt} dt = \Laplace[a(t) t](w) \ .$$
\end{proof}

\begin{lemma}  \label{cumbounds:lemma}
For all $T>0$, 
$$\LapM^{-1}[\softCap_T(w)](t) =    T \delta_{1/T}(t)\ ,$$ where $\delta_{1/T}$ is 
the Dirac Delta function  at $1/T$.  
 For    all  $p\in(0,1)$,  
$$\LapM^{-1}[w^p](t) = \frac{p}{\Gamma(1-p)} t^{-(1+p)} \,$$
where $\Gamma$ is the Gamma function.  
$$\LapM^{-1}[\ln(1+w)](t)  = \frac{e^{-t}}{t}\ .$$
 \end{lemma}
\begin{proof}
The claim for the inverse transform of $f(w) = \softCap_T(w)$ follows 
directly from \eqref{softcapdef:eq}. 
For $f(w)=w^p$, we 
apply Lemma \ref{invlap:eq} using $\frac{\partial w^p}{\partial w} = pw^{p-1}$. 
We take the inverse Laplace transform to obtain 
\begin{eqnarray*}
a(t) &=&  \frac{p}{t} \Laplace^{-1}[w^{p-1}](t) \\
 &=& \frac{p}{t} \frac{t^{-p}}{\Gamma(1-p)} \\
 &=& \frac{p}{\Gamma(1-p)}  t^{-(1+p)} \ . 
\end{eqnarray*}
We now consider 
$f(w) = \ln(1+w)$ and apply Lemma~\ref{invlap:eq} substituting the 
derivative $\frac{\partial f(w)}{\partial w} = \frac{1}{1+w}$  and 
$\Laplace^{-1}[1/(w+1)] = e^{-t}$.  We obtain 
$$a(t) =  \frac{1}{t} \Laplace^{-1}[1/(w+1)](t) = \frac{e^{-t}}{t}\ .$$
\end{proof}
% In particular,  the moment function $f(w)=\sqrt{w}$, we have inverse transform $\LapM^{-1}[\sqrt{w}](t)  = \frac{1}{2\sqrt{\pi}} t^{-1.5}$. 

\section{Double-logarithmic size max-distinct sketches} \label{mxdistinctcompact:sec}
We outline here basic representation tricks that yield max-distinct sketches of size  $O(\log\log n +
\epsilon^{-2}\log\epsilon^{-1})$ (assuming $1\leq m_x = O(\text{poly}(n))$.) 
Our sketch contains $k$ registers, where each register corresponds to 
a balanced ``bucket'' of output keys.  For balance, we can place all keys in 
all buckets (with different hash function per bucket) or (see Section~\ref{practicalopt:sec}) apply stochastic 
averaging and partition the 
output elements to buckets according to a partition of $[r]$. 
 For an 
element $e$ with key that falls in the bucket, we compute 
$-\ln u_{\text{\em{e.key}}} /\text{\em{e.value}}$ and if smaller than the current register,
replace its value.  The 
distribution of the minimum in the bucket is exponential with parameter equals to 
the sum of $m_x$ over keys in the bucket (see e.g. \cite{ECohen6f}). 
As in Hyperloglog, we only store exponents with base $b=1+\delta$ (for 
some appropriate small fixed $\delta$) 
 We apply consistent (per key/bucket) randomized rounding to an 
integral power of $(1+\delta)$ and store only the negated exponent 
$y_i$ for bucket $i$.  The exponent can be stored in $O(\log\log n)$ bits.   For different 
buckets, we can store one exponent and offsets of expected size $O(1)$
per bucket.  To obtain the approximate statistics from the sketch we 
use the estimator 
$$\widehat{\MdCount(E)} = (k-1)/\sum_{i=1}^k b^{-y_i}\ $$ 
for the total weight of keys over buckets (when stochastic averaging 
is not used we need to divide by $k$.) 

 \subsection{Sidelined output 
   elements} \label{sidelined:sec}
We discuss here compact representation of the $\ell=3\epsilon^{-2}$
``sidelined'' output keys and their $y$ values $(x_i,y_i)$.
Roughly, it suffices to store $O(\log \epsilon)$ per sidelined key (and one 
exponent):  An $O(\log \ell)$ representation of a hash of the key that will allow us to tell it apart (with good probability) from new elements with $y$ value
in the sidelined range.
An offset rounded representation of 
the $y$ value (so that $\int_y^\infty a(t)dt$  is approximated within 
$\epsilon$), and a $\log\epsilon$ representation of the distinct counter
bucket and of the random hash $u_e$,
so that the modification to the approximate counter can be performed 
when $(x_i,v_i)$ leaves the sidelined set.
% In all, we have sketch of size 
%$O(\epsilon^{-2}\log \epsilon + \log\log \MAX\{\SUM(W),n \int_\tau^0 
% a(t) dt\}$. 

\section{Practical optimization for point measurements} \label{practicalopt:sec}
An optimization that applies with 
approximate distinct counters (including Hyperloglog) that use stochastic 
averaging is to partition equally the $r$ ``output'' buckets to the 
$k$ counter  buckets.  Then pipe corresponding outkeys to each bucket. In more detail, stochastic averaging counters hash each key to one of $k$ buckets, where each bucket maintains an order 
  statistics (maximum or minimum) of another hash applied to the keys 
  that fall in that bucket. This design can be integrated 
with our element processing parameter $r$.  We 
direct the $r$ indices of outkeys produced for 
  same element to corresponding buckets in the counter, replacing the 
  hash-based assignment to buckets. 
%  Combining Lemma~\ref{smallr:lemma} with quality 
%  analysis of approximate distinct counters, using 
%  $r=O(\epsilon^{-2})$ and an approximate counter with $r$ registers 
%  will provide us with an estimator $\widehat{\widehat{\LapM}}[W](t)$ that has 
%  normalized root mean square error $\epsilon$ and good concentration. 

\section{All-threshold sketches} \label{ATsketch:sec}

 At a high level,
 sample-based distinct count sketches can be extended to all-threshold sketches by considering the distinct sketch of $\Okeys_t$ as a function of $t$.  It turns out (analysis of all-distance sketches) that the number of points $t$ where the sketch changes as we sweep $t$ is in expectation $\epsilon^{-2}\ln (rn)$.  In more detail, the Hyperloglog sketch consists of $k=\epsilon^{-2}$ registers that can only increase as keys are processed.  Each register has in expectation $\ln (rn)$ change points as we sweep $t$.

The storage overhead factor of all-threshold sketching
is the number of change points of each register ($\ln n$) multiplied 
by the representation of each $t$ value where the change occurred.
We now note that since $\LapM[W](t)$ is smooth and Lipschitz, a
small number of $\log(1/\epsilon)$ significant bits and the exponent suffices.
Since the change points for each register are an increasing sequence, generally from a  polynomial range, it suffices to store changes to the exponent.
% Exponents of both $t$ values and register 
%   values can be stored using offsets from previous breakpoints.  We 
%  can also still use offsets to compactly store base register contents. 

\ignore{
  We now consider working with multiple values of $r$ to make element 
  processing more efficient.  We partition $t$ to ranges where in each range we use a 
 particular value of $r$. 
 An upper bound on the number of ranges of $r$ values we need to use is 
$\log  \frac{m \max_e \text{\em{e.value}}}{\min_e \text{\em{e.value}}}\ ,$ where $m$ is the 
number of elements.    We will assume this value is $O(\log n)$. 
}
% To summarize, with overhead of $O(\log^2 n)$, we can maintain 
% approximate $\LapM$ measurements for all values of $t$. 

% This means that the bottleneck in obtaining approximation quality for 
% the Laplace$^c$ measurements is the quality of the approximate distinct 
% counters 
% we apply to the outkeys. 

%   structure of the 
% To obtain the normalized Laplace measurements, we also use another 
% counter to count distinct $\text{\em{e.key}}$ (this is equivalent to using 
% $t=\infty$).  
% Note that This count will always be at least 
% the count for any $t$

 With an enhanced sketch, since it records all counts as we sweep $t$, we can 
apply the HIP estimator~\cite{ECohenADS:TKDE2015} to the sketch (with 
respect to the swept parameter $t$) to   obtain full-range estimates 
for all $t$. 
\ignore{
  The representation of the approximate transform 
  $\widehat{\widehat{\LapM}}[W](t)$ we obtain using the approximate 
  threshold counts  is piecewise 
  linear non-decreasing monotone function with a logarithmic number of 
breakpoints.    The first piece is linear and has the form $t \sum_x w_x$ and 
other pieces are constant.  The last piece has the 
approximate distinct count $\hat{n}$ (see discussion in 
Section~\ref{rchoice:sec} on the relevant 
range). 
When we use the parameter $r = C\epsilon^{-2}$, the slope of the first 
piece is either approximated from a (straightforward) separate computation of the sum 
$\sum_x w_x$ or from the smallest $t$ for which the total approximate 
count 
$\TdCount(\Oelems(W))$ exceeded $\Omega(C \epsilon^{-2})$. 
}

\section{Bin Transform} \label{bintrans:sec}
A related discrete transform is the 
{\em Bin transform}   $\BinM[W](\ell)$ for an integer parameter 
$\ell > 0$ is defined as 
$$\BinM[W](\ell) \equiv \int_0^\infty  W(w) \ell(1-(1-1/\ell)^w) dw\ .$$  
  We present a mapping scheme of elements to outkeys that for parameters $\ell,r$ has expected 
  distinct count equal to the Bin trasform of the frequencies.  The mapping 
applies for  data sets of elements with uniform values $\text{\em{e.value}} =1$. 
 We use $r$ sets of $\ell$ independent hash functions $H^{(j)}_i$ for $i\in [\ell]$ and $j \in [r]$. 
 The hash functions are applied to element keys and produces output keys.  We 
  assume that $H^{(j)}_i(x)$ are unique for all $x,i,j$. 
  When processing an element $e$, we draw for each $j\in [r]$, a value 
  $i \sim U[\ell]$ uniformly at random.  We generate the output key 
$H^{(j)}_i(\text{\em{e.key}})$.  Note that $r$ output keys are generated for each 
element. 
Our measurement $\widehat{\BinM}[W](\ell)$ is the number of distinct 
outkeys divided by $r$.

 The expected contribution of a key $x$ with weight $w_x$ to the 
 measurement (the expected number of distinct outkeys produced for 
 elements with that key divided by $r$) is $\ell(1-(1-1/\ell)^w)$
 which is close to $\LapM$measurements.  
 We focus more on the continuous $\LapM$ transform, however, as it is nicer to work with. 
A variation of the Bin transform was used in \cite{freqCap:KDD2015} for 
weighted sampling of keys and it was also proposed to use them for counting 
(computing statistics over the full data).  

\section{Extension: Time-decaying aggregation} \label{decay:sec}

We presented simple ``reductions'' from the approximation 
 of (sub)linear growth statistics to distinct or max-distinct counting and sum. 
 We observe that our reductions can also be combined with time-decaying aggregation \cite{CoSt:pods03f}
 (e.g., sliding window models), with essentially out-of-the-box use of respective structures 
for distinct counting and sum.  For our enhanced (max-distinct,
all-threshold) counters we can obtain time-decaying aggregations for 
all decay functions by essentially ``distance sketching'' the time
domain (see \cite{ECohenADS:TKDE2015}).

\section{Extension: Weighted keys} \label{weighted:sec}

  A useful extension is to approximate
aggregations of the form  
\begin{equation} \label{weightedkeysagg:eq}
\sum_x f(w_x) v_x\ ,
\end{equation}
 where $v_x > 0$ are intrinsic weight associated 
with key $x$. 
With the distribution interpretation, we have $W(w)$ as the weighted
sum $$W(w) = \sum_{x \mid w_x=w} v_x$$ of all keys with frequency $w_x=w$.
Our presentation focused on $v_x =1$ for all keys.
We outline how our 
algorithms can be extended to handle general $v_x$ weights.
 The assumption here is that $v_x$ is available to us when an
element with key $x$ is processed.

 One application is computing a loss function or its gradient  when
 computing embeddings.  The
 data elements have entity and context pairs $(i,j)$.  Each entity $i$
 has a current embedding vector $\boldmath{v}_i$ and similarly, contexts $j$ have
 embeddings $\tilde{\boldmath{v}}_j$.    The pairs are our keys and their values
 are a function of the inner produce $\boldmath{v}_i\cdot \tilde{\boldmath{v}}_j$.   The weight
 $w_x$ of a key $x=(i,j)$ is the number of occurrences of the key in
 the data.   We would like to estimate \eqref{weightedkeysagg:eq} efficiently.

 When $f$ is a soft capping function, this can be done with a single modified
 point measurement.  For $f\in \overline{\softCap}$, we need a
 modified combination measurement.

  The modification for point measurements (Element processing
  Algorithm~\ref{elemmap1:alg}) 
produces output elements $(H_i(x),v_x)$ 
instead of respective output keys.  The measurement we seek is then a
a {\em weighted distinct statistics} of the output elements, which
returns the sum $\sum_y v_y$ over distinct output element keys
(divided by $r$). 
The measurement is approximated using an
approximate {\em weighted} distinct counter instead of a plain
distinct counter.

Weighted distinct statistics are a special case of max-distinct 
statistics, where elements with a certain key always have the same 
value.  Therefore, they  can be approximated using approximate max-distinct 
counters (see Section~\ref{maxdistinctcount:sec}).

The modification of Algorithm~\ref{mxdistinct:algo} needed to obtain a
combination measurements of \eqref{weightedkeysagg:eq} 
replaces the value $y$ of the output elements generated for an element 
with key $x$ by the product $v_x y$.  The elements are processed as 
before by an approximate max-distinct counter. The quality and
structure size tradeoffs remain the same.

\end{document}